\documentclass[11pt]{article}
\title{Improved bounds and new techniques for Davenport--Schinzel
sequences and their generalizations\footnote{A preliminary version
of this article appears in \emph{Proceedings of the 20th Annual
ACM-SIAM Symposium on Discrete Algorithms (SODA'09)} (New York, NY), pp.~1--10,
ACM and SIAM, 2009.}}
\author{Gabriel Nivasch\footnote{Institute of Theoretical Computer Science, ETH Z\"urich,
CH-8092 Z\"urich, Switzerland, \texttt{gabriel.nivasch@inf.ethz.ch}. Work was done while
the author was at Tel-Aviv University, and was supported by ISF Grant 155/05
and by the Hermann Minkoswki--MINERVA Center for Geometry at Tel
Aviv University.}}
\date{}

\usepackage{amsmath,amsfonts,amssymb,amsthm}
\usepackage{hyperref}
\usepackage{graphics}

\newtheorem{theorem}{Theorem}[section]
\newtheorem{lemma}[theorem]{Lemma}
\newtheorem{recurrence}[theorem]{Recurrence}
\newtheorem{conjecture}[theorem]{Conjecture}
\newtheorem{corollary}[theorem]{Corollary}

\newenvironment{definition}[1][Definition]
{\refstepcounter{theorem}\begin{trivlist}\item\textbf{#1
\thetheorem:\ }} {\end{trivlist}}

\newenvironment{remark}[1][Remark]
{\refstepcounter{theorem}\begin{trivlist}\item\textbf{#1
\thetheorem:\ }} {\end{trivlist}}

\newcommand\Ex{{\mathrm{Ex}}} % max length of gralized DS seq
\newcommand\F{{F}} % max length of formation-free seq
\newcommand\ADS{{\mathrm{ADS}}} % variant of DS seq used in alt. method
\newcommand\N{{\Pi}} % max # dist. symbols in variant DS seq
\newcommand\psif{{\psi'}} % psi for formation-free seq
\newcommand\AFF{{\mathrm{AFF}}} % variant of formation-free seq
\newcommand\Nf{{\Pi'}} % max # dist. symbols in variant FF seq

\begin{document}

\maketitle

\begin{abstract}
We present several new results regarding $\lambda_s(n)$, the maximum
length of a Davenport--Schinzel sequence of order $s$ on $n$
distinct symbols.

First, we prove that $\lambda_s(n) \le n \cdot 2^{(1/t!)\alpha(n)^t
+ O\left( \alpha(n)^{t-1}\right)}$ for $s\ge 4$ even, and
$\lambda_s(n) \le n \cdot 2^{(1/t!)\alpha(n)^t \log_2 \alpha(n) +
O\left(\alpha(n)^t \right)}$ for $s\ge 3$ odd, where $t = \lfloor
(s-2) / 2 \rfloor$, and $\alpha(n)$ denotes the inverse Ackermann
function. The previous upper bounds, by Agarwal, Sharir, and Shor
(1989), had a leading coefficient of $1$ instead of $1/t!$ in the
exponent. The bounds for even $s$ are now tight up to lower-order
terms in the exponent. These new bounds result from a small
improvement on the technique of Agarwal et al.

More importantly, we also present a new technique for deriving upper
bounds for $\lambda_s(n)$. This new technique is very similar to the
one we applied to the problem of stabbing interval chains (Alon et
al., 2008). With this new technique we: (1) re-derive the upper
bound of $\lambda_3(n) \le 2n \alpha(n) + O{\bigl(n \sqrt{\alpha(n)}
\bigr)}$ (first shown by Klazar, 1999); (2) re-derive our own new
upper bounds for general $s$; and (3) obtain improved upper bounds
for the generalized Davenport--Schinzel sequences considered by
Adamec, Klazar, and Valtr (1992).

Regarding lower bounds, we show that $\lambda_3(n) \ge 2n \alpha(n)
- O(n)$ (the previous lower bound (Sharir and Agarwal, 1995) had a
coefficient of $1\over 2$), so the coefficient $2$ is tight. We also
present a simpler variant of the construction of Agarwal, Sharir,
and Shor that achieves the known lower bounds of $\lambda_s(n) \ge n
\cdot 2^{(1/t!) \alpha(n)^t - O\left(\alpha(n)^{t-1}\right)}$ for
$s\ge 4$ even.

\end{abstract}

\section{Introduction}
Given a sequence $S$, denote by $|S|$ the length of $S$, and by
$\|S\|$ the number of distinct symbols in $S$. If $u$ is another
sequence, we write $u \subset S$ if $S$ contains a subsequence $u'$
(not necessarily contiguous) which is isomorphic to $u$ (i.e., $u'$
can be made equal to $u$ by a one-to-one renaming of its symbols).
In this case we say that $S$ \emph{contains} $u$ or that $u$
\emph{is contained} in $S$. Otherwise, we write $u\not\subset S$ and
we say that $S$ is \emph{$u$-free}. For example, $S = abcdbc$
contains $u = abab$, but it is $v$-free for $v = abba$.

A sequence $S = a_1 a_2 a_3 \ldots$ is called \emph{$r$-sparse} if
$a_i \neq a_j$ whenever $1\le |j-i| \le r-1$. In other words, $S$ is
$r$-sparse if every interval in $S$ of length at most $r$ contains
only distinct symbols.

A \emph{Davenport--Schinzel sequence of order $s$}, for $s\ge 1$, is
a sequence that is $2$-sparse (i.e., contains no adjacent repeated
symbols) and is $u$-free for $u=ababab\ldots$ of length $s+2$. In
other words, a Davenport--Schinzel sequence of order $s$ does not
contain any alternation $a\ldots b\ldots a\ldots b\ldots$ of length
$s+2$ for any pair of symbols $a$, $b$.

Let $\lambda_s(n)$ denote the maximum length of a
Davenport--Schinzel sequence of order $s$ on $n$ distinct symbols
($\lambda_s(n)$ is finite for all $s$ and $n$). We always take $s$
to be fixed, and consider $\lambda_s(n)$ as a function of $n$.

These sequences are named after Harold Davenport and Andrzej
Schinzel, who first studied them in 1965 \cite{DS_original}. The
main motivation for Davenport--Schinzel sequences is the complexity
of the lower envelope of a set of curves in the plane. However,
Davenport--Schinzel sequences have a large number of applications in
computational and combinatorial geometry; the book \cite{DS_book} by
Sharir and Agarwal is entirely devoted to this topic. Given the
prominent role these sequences play in computational geometry, it is
of great interest to derive tight asymptotic bounds for
$\lambda_s(n)$. This goal is quite challenging, given the
complicated form of the known bounds (see below). There has been
little progress in the problem for nearly 20 years.

The bounds $\lambda_1(n) = n$ (no $aba$) and $\lambda_2(n) = 2n-1$
(no $abab$) are quite easy to obtain. But for $s\ge 3$ the problem
becomes much more complicated---it turns out that $\lambda_s(n)$ is
slightly superlinear in $n$.

Hart and Sharir showed in 1986 \cite{HS86,DS_book} that
$\lambda_3(n) = \Theta(n \alpha(n))$, where $\alpha(n)$ denotes the
inverse Ackermann function. (For the upper bound see also Sharir
\cite{sharir87} and Klazar \cite{klazar99}, and for the lower bound
see also Wiernik and Sharir \cite{WS88}, Komj\'ath \cite{komjath},
and Shor \cite{shorI}.)

The tightest known bounds for $\lambda_3(n)$ are
\begin{equation}\label{eq_lambda_3_old_sandwich}
{1\over 2} n \alpha(n) - O(n) \le \lambda_3(n) \le 2n \alpha(n) +
O{\bigl( n \sqrt{\alpha(n)} \bigr)}.
\end{equation}
The lower bound is due to Sharir and Agarwal~\cite{DS_book} (based
on the construction by Wiernik and Sharir~\cite{WS88}). The upper
bound is due to Klazar~\cite{klazar99}. Klazar~\cite{klazar_survey}
asks whether $\lim_{n\to\infty} \lambda_3(n) / \bigl( n \alpha(n)
\bigr)$ exists.

The current upper and lower bounds for $\lambda_s(n)$ for general
$s$ were established by Agarwal, Sharir, and Shor in 1989
\cite{ASS89,DS_book}, and are as follows. Let $t = \lfloor (s-2) / 2
\rfloor$. Then,
\begin{align}
\lambda_s(n) &\le
\begin{cases}
n \cdot 2^{\alpha(n)^t + O\left( \alpha(n)^{t-1}\right)},
& \text{$s\ge 4$ even};\\
n \cdot 2^{\alpha(n)^t \log_2 \alpha(n) + O\left(\alpha(n)^t
\right)}, & \text{$s\ge 3$ odd};
\end{cases}\nonumber\\
\lambda_s(n) &\ge n\cdot 2^{(1/t!) \alpha(n)^t -
O\left(\alpha(n)^{t-1} \right)}, \quad s\ge 4 \text{
even}.\label{eq_lambda_lower}
\end{align}
For odd $s\ge 5$ the asymptotically best lower bounds known are
obtained by $\lambda_s(n) \ge \lambda_{s-1}(n)$.

Sharir and Agarwal's book \cite{DS_book} contains a complete
derivation of the current upper and lower bounds for $\lambda_s(n)$
for all $s$.

In 2008 the author, together with Alon, Kaplan, Sharir, and
Smorodinsky, conjectured that:
\begin{conjecture}[\cite{interval_chains}]\label{conjecture_lambda}
The true bounds for $\lambda_s(n)$ are
\begin{equation*}
\lambda_s(n)=
\begin{cases}
n \cdot 2^{(1/t!)\alpha(n)^t \pm
O\left(\alpha(n)^{t-1}\right)}, & \text{$s\ge 4$ even};\\
n\cdot 2^{(1/t!)\alpha(n)^t \log_2 \alpha(n) \pm
O\left(\alpha(n)^t\right)}, & \text{$s\ge 3$ odd};
\end{cases}
\end{equation*}
where $t = \lfloor (s-2)/2 \rfloor$.
\end{conjecture}
This conjecture is based on some surprisingly similar tight bounds
that they obtained for an unrelated problem called \emph{stabbing
interval chains with $j$-tuples}.

\subsection{Generalized Davenport--Schinzel sequences}

Adamec, Klazar, and Valtr~\cite{AKV92} considered a generalization
of Davenport--Schinzel sequences, in which the forbidden pattern is
not limited to $abab\ldots$, but can be an arbitrary sequence.

Let $u$ (the \emph{forbidden pattern}) be a sequence with $\|u\| =
r$ distinct symbols and length $|u| = s$. Then we denote by
$\Ex_u(n)$ the maximum length of an $r$-sparse, $u$-free sequence on
$n$ distinct symbols. The standard Davenport--Schinzel sequences are
obtained by taking $r=2$ and $u = abab\ldots$ of length $s+2$.

The requirement of $r$-sparsity is necessary, since an
$(r-1)$-sparse, $u$-free sequence can be arbitrarily long. The
requirement of $r$-sparsity, however, ensures that $\Ex_u(n)$ is
finite.

Generalized Davenport--Schinzel sequences have found several
applications in discrete mathematics. Valtr~\cite{valtr_survey} used
generalized Davenport--Schinzel sequences to obtain bounds for some
Tur\'an-type problems for geometric graphs. Alon and Friedgut
\cite{alon_friedgut} used them to derive an almost-tight upper bound
for the so-called Stanley--Wilf conjecture (the conjecture was later
proved by Marcus and Tardos~\cite{marcus_tardos} by a different
technique). For more information see the surveys by Klazar
\cite{klazar_survey} and by Valtr \cite{valtr_survey}. More
recently, Pettie~\cite{pettie} used generalized Davenport--Schinzel
sequences to improve Sundar's~\cite{sundar} near-linear upper bound
for the \emph{deque conjecture} for splay trees.

\subsection{Formation-free sequences}

Klazar in 1992 \cite{klazar92} developed a general technique for
bounding $\Ex_u(n)$ in terms of only $r = \|u\|$ and $s = |u|$. His
technique is based on considering what we call \emph{formation-free
sequences} (our name). Given integers $r$ and $s$, an
\emph{$(r,s)$-formation} is a sequence of $s$ permutations on $r$
symbols. For example, $abcd\allowbreak\ dcab\allowbreak\
dcab\allowbreak\ cdba\allowbreak\ dabc$ is a $(4,5)$-formation. An
\emph{$(r,s)$-formation-free sequence} is a sequence which is
$r$-sparse and does not contain any $(r,s)$-formation as a
subsequence.

Denote by $\F_{r,s}(n)$ the length of the longest possible
$(r,s)$-formation-free sequence on $n$ distinct symbols. Let $u$ be
a sequence with $\|u\| = r$ and $|u| = s$. Since $u$ is trivially
contained in every $(r,s)$-formation, it follows that $\Ex_u(n) \le
\F_{r,s}(n)$.

Klazar made a slight improvement to this observation, by noting that
if $r\ge 2$, then $u$ is contained in every $(r,s-1)$-formation, and
thus,
\begin{equation}\label{eq_Klazar_Ex_F}
\Ex_u(n) \le \F_{r,s-1}(n) \qquad \text{for $r\ge 2$.}
\end{equation}
(The case $r=1$ is not interesting in any case.) Klazar proved the
bound
\begin{equation*}
\F_{r,s}(n) \le n \cdot 2^{O{\left( \alpha(n)^{s-3} \right)}},
\end{equation*}
where the $O$ notation hides constants that depend on $r$ and $s$.
Together with (\ref{eq_Klazar_Ex_F}), this implies that
\begin{equation*}
\Ex_u(n) \le n \cdot 2^{O{\left(\alpha(n)^{s-4} \right)}}.
\end{equation*}

\subsection{Our results}

In this paper we present several new results.

First, we make a small improvement on the argument of Agarwal et
al.~\cite{ASS89,DS_book} and prove:

\begin{theorem}\label{thm_lambda_new_upper}
Let $s\ge 3$ be fixed, and let $t = \lfloor (s-2) / 2 \rfloor$. Then
\begin{equation*}
\lambda_s(n) \le
\begin{cases}
n \cdot 2^{(1/t!)\alpha(n)^t + O\left(\alpha(n)^{t-1}\right)}, &
\text{$s$ even};\\
n \cdot 2^{(1/t!) \alpha(n)^t \log_2 \alpha(n) + O\left(\alpha(n)^t
\right)}, & \text{$s$ odd}.
\end{cases}
\end{equation*}
\end{theorem}

Thus, the upper bounds for $\lambda_s(n)$ are now in line with
Conjecture~\ref{conjecture_lambda}, and for $s$ even they are also
tight up to lower-order terms in the exponent.

More importantly, we also present a new technique for deriving upper
bounds for $\lambda_s(n)$. Our new technique is based on some
recurrences very similar to those used by Alon et
al.~\cite{interval_chains}, for the problem of stabbing interval
chains with $j$-tuples.

With our new technique we re-derive Klazar's upper bound
(\ref{eq_lambda_3_old_sandwich}) for $\lambda_3(n)$, as well as our
new bounds in Theorem~\ref{thm_lambda_new_upper} for $\lambda_s(n)$,
$s\ge 4$. We also apply our technique to formation-free sequences,
proving that:

\begin{theorem}\label{thm_formation_free_upper}
For $s\ge 4$ we have
\begin{equation*}
\F_{r,s}(n) \le
\begin{cases}
n \cdot 2^{ (1/t!) \alpha(n)^t + O{\left( \alpha(n)^{t-1} \right)}},
&
\text{ $s$ odd};\\
n \cdot 2^{ (1/t!) \alpha(n)^t \log_2 \alpha(n) + O{\left(
\alpha(n)^t \right)}}, &
\text{ $s$ even};\\
\end{cases}
\end{equation*}
where $t = \lfloor (s-3) / 2 \rfloor$. (The $O$ notation hides
factors dependent on $r$ and $s$.)
\end{theorem}

As an aside, we improve on Klazar's bound (\ref{eq_Klazar_Ex_F}):

\begin{lemma}\label{lemma_Ex_to_F}
Let $u$ be a sequence with $\|u\| = r$, $|u| = s$. Then, $\Ex_u(n)
\le \F_{r,s-r+1}(n)$.
\end{lemma}

This, together with Theorem~\ref{thm_formation_free_upper},
yields:\footnote{Klazar himself~\cite{klazar92} speculated that it
should be possible to achieve roughly $\Ex_u(n) \le n\cdot
2^{O\bigl(\alpha(n)^{s/2}\bigr)}$.}

\begin{theorem}\label{thm_Ex_upper}
Let $u$ be a sequence with $\|u\| = r$, $|u| = s$, and $s\ge r+3$.
Let $t = \lfloor (s-r-2) / 2 \rfloor$. Then,
\begin{equation*}
\Ex_u(n)\\ \le
\begin{cases}
n \cdot 2^{(1/t!)\alpha(n)^t + O\left( \alpha(n)^{t-1}\right)}, &
\text{ $s-r$ even};\\
n \cdot 2^{(1/t!)\alpha(n)^t \log_2 \alpha(n) + O\left(\alpha(n)^t
\right)}, & \text{ $s-r$ odd}.
\end{cases}
\end{equation*}
\end{theorem}

Note that Theorem~\ref{thm_Ex_upper} is a generalization of
Theorem~\ref{thm_lambda_new_upper}: Taking $r = 2$ and $u =
abab\ldots$ of length $s+2$ yields the theorem once again.

Regarding lower bounds, we prove:

\begin{theorem}\label{thm_lambda3_lower}
$\lambda_3(n) \ge 2n \alpha(n) - O(n)$.
\end{theorem}
\begin{corollary}\label{cor_lim_lambda_3}
$\lim_{n\to\infty} \lambda_3(n) / \bigl(n \alpha(n) \bigr) = 2$.
\end{corollary}

Finally, we present a simpler variant of the construction of
Agarwal, Sharir, and Shor \cite{ASS89,DS_book}, which achieves the
lower bounds (\ref{eq_lambda_lower}) for $s\ge 4$ even.

\subsection{The Ackermann function and its inverse}

Let us define (our version of) the Ackermann function and its
inverse.

The \emph{Ackermann hierarchy} is a sequence of functions $A_k(n)$,
for $k = 1, 2, 3, \ldots$ and $n\ge 0$, where $A_1(n) = 2n$, and for
$k\ge 2$ we let $A_k(n) = A_{k-1}^{(n)}(1)$. (Here $f^{(n)}$ denotes
the $n$-fold composition of $f$.) Alternatively, the definition of
$A_k(n)$ for $k\ge 2$ can be written recursively as
\begin{equation}\label{eq_def_Ak}
A_k(n) =
\begin{cases}
1, & \text{if $n=0$};\\
A_{k-1}{\bigl( A_k(n-1) \bigr)}, & \text{otherwise}.
\end{cases}
\end{equation}
We have $A_2(n) = 2^n$, and $A_3(n) = 2^{2^{\cdots^2}}$ is a
``tower" of $n$ twos. Each function in this hierarchy grows much
faster than the preceding one. Namely, for every fixed $k$ and $c$
we have $A_{k+1}(n) \ge A_k^{(c)}(n)$ for all large enough $n$.

Notice that $A_k(1) = 2$ and $A_k(2) = 4$ for all $k$, but $A_k(3)$
already grows very rapidly with $k$. We define the \emph{Ackermann
function} as $A(n) = A_n(3)$. Thus, $A(n) = 6, 8, 16, 65536, \ldots$
for $n = 1, 2, 3, 4, \ldots$.\footnote{The Ackermann function is
usually defined by ``diagonalizing" the hierarchy, letting $A'(n) =
A_n(n)$. This does not make any asymptotic difference, since
$A'(n-2) \le A(n) \le A'(n-1)$ for $n\ge 5$. (There are several
other definitions of the Ackermann hierarchy and function in the
literature, all of which exhibit equivalent rates of growth.) We
prefer the above definition because, first, ``diagonalization" is
unnecessary, and second, the corresponding definition
(\ref{eq_def_alpha}) of $\alpha(x)$ comes out simpler. For other
references where $\alpha(x)$ is defined as in (\ref{eq_def_alpha})
see Pettie \cite{pettie} and Seidel \cite[slide 85]{seidel_pdf}.}

For every fixed $k$ we have $A(n) \ge A_k(n)$ for all large enough
$n$. It is also easy to verify that
\begin{equation}\label{eq_another_rec_A}
A(n) = A_{n-2}{\bigl( A(n-1) \bigr)}, \qquad \text{for $n\ge 3$}.
\end{equation}

We then define the slow-growing inverses of these rapidly-growing
functions as
\begin{align}
\alpha_k(x) &= \min{\{ n \mid A_k(n) \ge x \}}, \label{eq_alphak_Ak}\\
\alpha(x) &= \min{\{ n \mid A(n) \ge x\}} \label{eq_alpha_A},
\end{align}
for all real $x\ge 0$.

Alternatively, and equivalently, we can define these inverse
functions directly without making reference to $A_k$ and $A$. We
define the \emph{inverse Ackermann hierarchy} by $\alpha_1(x) =
\lceil x/2 \rceil$ and
\begin{equation}\label{eq_rec_alpha_k}
\alpha_k(x) =
\begin{cases}
0, & \text{if $x\le 1$};\\
1+ \alpha_k{\bigl( \alpha_{k-1}(x) \bigr)}, & \text{otherwise};
\end{cases}
\end{equation}
for $k\ge 2$. In other words, for each $k\ge 2$, $\alpha_k(x)$
denotes the number of times we must apply $\alpha_{k-1}$, starting
from $x$, until we reach a value not larger than $1$. Thus,
$\alpha_2(x) = \lceil \log_2 x \rceil$, and $\alpha_3(x) = \log^*x$.
Finally, we define the \emph{inverse Ackermann function} by
\begin{equation}\label{eq_def_alpha}
\alpha(x) = \min{\{ k \mid \alpha_k(x) \le 3 \}}.
\end{equation}

It is an easy exercise (only slightly tedious) to prove by induction
that the above two definitions of $\alpha_k$ and $\alpha$ are
exactly equivalent.

\subsection{Organization of this paper}

Sections \ref{sec_DS_gral_upper}--\ref{sec_FF} contain our
upper-bound results. In Section~\ref{sec_DS_gral_upper} we show how
Theorem~\ref{thm_lambda_new_upper} reduces to bounding a function
denoted $\psi_s(m,n)$. In Section~\ref{sec_psi_old} we improve the
technique of Agarwal et al.~\cite{ASS89,DS_book} for bounding
$\psi_s(m,n)$. In Section~\ref{sec_psi_new} we present an
alternative technique, which yields the same improved bounds for
$\psi_s(m,n)$.

Section~\ref{sec_FF} addresses formation-free sequences. We first
prove Lemma~\ref{lemma_Ex_to_F}, and then we extend our new
technique to formation-free sequences, proving
Theorem~\ref{thm_formation_free_upper}.

Sections \ref{sec_constr_3}--\ref{sec_constr_even} contain our
lower-bound results. Section~\ref{sec_constr_3} presents our
construction for $\lambda_3(n)$ that proves
Theorem~\ref{thm_lambda3_lower}. Section~\ref{sec_constr_even}
contains our simplified construction of Davenport--Schinzel
sequences of even order $s\ge 4$.

Appendices \ref{app_lambda_3_klazar}--\ref{app_ack_like} contain
some technical calculations.

For completeness, we provide proofs in this paper of most of the
previous results we rely on.

\section{Upper bounds for Davenport--Schinzel sequences}
\label{sec_DS_gral_upper}

The upper bounds for $\lambda_s(n)$ are obtained by considering a
function with an additional parameter $m$:

\begin{definition}\label{def_psi}
Let $m$, $n$, and $s$ be positive integers. Then $\psi_s(m,n)$
denotes the maximum length of a Davenport--Schinzel sequence of
order $s$ on $n$ distinct symbols that can be partitioned into $m$
or fewer contiguous blocks, where each block contains only distinct
symbols.
\end{definition}

The relation between $\lambda_s(n)$ and $\psi_s(m,n)$ is as follows:

\begin{lemma}[\cite{ASS89, DS_book}]\label{lemma_lambda_to_psi}
Let $\varphi_{s-2}(n)$ be a nondecreasing function in $n$ such that
$\lambda_{s-2}(n) \le n \varphi_{s-2}(n)$ for all $n$. Then,
\begin{equation}\label{eq_lambda_to_psi}
\lambda_s(n) \le \varphi_{s-2}(n) \bigl( \psi_s(2n,n) + 2n \bigr).
\end{equation}
\end{lemma}

\begin{proof}
Let $S$ be a Davenport--Schinzel sequence of order $s$ on $n$
symbols with maximum length $\lambda_s(n)$. Partition $S$ greedily
from left to right into blocks $S_1, S_2, \ldots, S_m$, such that
each $S_i$ is a sequence of order $s-2$; in other words, when
scanning $S$ from left to right, start a new block $S_{i+1}$ only if
an additional symbol would cause $S_i$ to contain an alternation of
length $s$.\footnote{This greedy left-to-right approach is in fact
optimal---it yields a partition of $S$ into the minimum possible
number of blocks of specified order $r<s$.} We claim that $m$, the
number of blocks, is at most $2n$.

Indeed, consider some block $S_i$ for $i<m$. Since $S_i$ was
extended maximally to the right, it must contain an alternation
$abab\ldots$ of length $s-1$, which is extended to length $s$ by the
first symbol of $S_{i+1}$ (which is either $a$ or $b$, depending on
the parity of $s$). But then, we cannot have \emph{both} $b$
appearing in a previous $S_j$, $j<i$, \emph{and} $b$ or $a$
(depending on the parity of $s$) appearing in a subsequent $S_j$,
$j>i$, because then $S$ would contain a forbidden alternation of
length $s+2$.

Hence, each block $S_i$ (including the last one $S_m$) contains
either the first occurrence or the last occurrence of at least one
symbol. Thus, $m\le 2n$.

Let $n_i = \| S_i\|$. Then,
\begin{equation*}
\lambda_s(n) = |S| = \sum_{i=1}^m |S_i| \le \sum_{i=1}^m
\lambda_{s-2}(n_i) \le \sum_{i=1}^m n_i \varphi_{s-2}(n_i) \le
\varphi_{s-2}(n) \sum_{i=1}^m n_i.
\end{equation*}

Let us now bound $\sum n_i$. Construct a subsequence $S'$ of $S$ by
taking, for each block $S_i$, just the first occurrence of each
symbol in $S_i$. Note that $S'$ has length $|S'| = \sum n_i$ and,
being a subsequence of $S$, it contains no alternation of length
$s+2$. Furthermore, $S'$ is decomposable into $m$ blocks of distinct
symbols $S'_1, \ldots, S'_m$. However, $S'$ might contain adjacent
equal symbols at the interface between blocks, but by removing at
most one symbol from each block $S'_i$, we can obtain a sequence
$S''$ with no adjacent equal symbols. Therefore, $|S''| \le
\psi_{s}(m,n)$, and so $|S'|\le \psi_{s}(m,n) + m$. Since $m\le
2n$, we conclude that
\begin{equation*}
\lambda_s(n) \le \varphi_{s-2}(n) \bigl( \psi_{s}(2n,n) + 2n
\bigr). \qedhere
\end{equation*}
\end{proof}

In particular, for $s=3$ we have $\lambda_3(n) \le \psi_3(2n,n) +
2n$ (by taking $\varphi_1(n) = 1$, since $\lambda_1(n) = n$).
Actually for $s=3$ we have $\lambda_3(n) = \psi_3(2n,n)$ (Hart and
Sharir~\cite{HS86,DS_book}).

The main issue, then, is to bound $\psi_s(m,n)$. We present two
different techniques for bounding $\psi_s(m,n)$. The first one is a
minor modification of the technique of Agarwal et
al.~\cite{ASS89,DS_book}. The second one is our new technique. Both
techniques yield the following bounds:

\begin{lemma}\label{lemma_gral_psi_alphak_bound}
For $s=3$ we have
\begin{equation*}
\psi_3(m,n) = O{\bigl( km\alpha_k(m) + kn \bigr)} \qquad \text{for
all $k$}.
\end{equation*}
In general, for every fixed $s\ge 3$ we have
\begin{equation*}
\psi_s(m,n) \le C_{s,k}\bigl( m \alpha_k(m)^{s-2} + n\bigr) \qquad
\text{for all $k$},
\end{equation*}
for some constants $C_{s,k}$ of the form
\begin{equation*}
C_{s,k} =
\begin{cases}
2^{(1/t!) k ^t \pm O{\left( k^{t-1} \right)}}, & \text{$s$ even};\\
2^{(1/t!) k ^t \log_2 k \pm O{\left( k^t \right)}}, & \text{$s$
odd};
\end{cases}
\end{equation*}
where $t = \lfloor (s-2) / 2 \rfloor$.
\end{lemma}
(Equivalent bounds for $\psi_3(m,n)$ and $\psi_4(m,n)$ were
previously derived by Hart and Sharir~\cite{HS86,DS_book}, and
Agarwal, Sharir, and Shor~\cite{ASS89,DS_book}, respectively. For
$s\ge 5$ these are improvements over \cite{ASS89,DS_book}, which for
$s\ge 6$ yield improved bounds for $\lambda_s(n)$.)

From Lemmas~\ref{lemma_lambda_to_psi}
and~\ref{lemma_gral_psi_alphak_bound} it follows that $\lambda_s(n)
= o(n \alpha_\ell(n))$ for every fixed $\ell$: Just take $k =
\ell+1$ in Lemma~\ref{lemma_gral_psi_alphak_bound}, bounding
$\varphi_{s-2}(n)$ in Lemma~\ref{lemma_lambda_to_psi} by induction.
Here the magnitude of the constants $C_{s,k}$ is irrelevant.

But we can also derive a tighter bound for $\lambda_s(n)$, namely
Theorem~\ref{thm_lambda_new_upper}, if we let $k$ grow very slowly
with $m$; for this the dependence of $C_{s,k}$ on $k$ is
significant:

\begin{proof}[Proof of Theorem~\ref{thm_lambda_new_upper}]
Take $k = \alpha(m)$ in Lemma~\ref{lemma_gral_psi_alphak_bound}
(recalling that $\alpha_{\alpha(m)}(m) \le 3$ by definition), and
substitute into Lemma~\ref{lemma_lambda_to_psi}. For $s = 3,4$ we
get $\lambda_3(n) = O{\left( n \alpha(n) \right)}$, $\lambda_4(n) =
O{\left(n \cdot 2^{\alpha(n)}\right)}$ (by taking $\varphi_1(n) =
1$, $\varphi_2(n) = 2$). For $s\ge 5$ we bound $\lambda_{s-2}(n)$ by induction on $s$, and we substitute the resulting bound for $\varphi_{s-2}(n)$ into (\ref{eq_lambda_to_psi}). We obtain the desired bounds (the factor
$\varphi_{s-2}(n)$ only affects lower-order terms in the exponent).
\end{proof}

\section{Bounding $\psi_s(m,n)$---improving the known technique}\label{sec_psi_old}

In this section we prove Lemma~\ref{lemma_gral_psi_alphak_bound} by
making a small improvement on the technique of Agarwal et
al.~\cite{ASS89,DS_book}. The main ingredient in the proof is the
following complicated-looking recurrence relation. This is a small
modification of the recurrence in \cite{ASS89,DS_book} (and more
complicated).

\begin{recurrence}\label{recurrence_psi}
Let $m,n \ge 1$ and $b\le m$ be integers, and let
\begin{equation*}
m = m_1 + m_2 + \cdots + m_b
\end{equation*}
be a partition of $m$ into $b$ nonnegative integers. Then, there
exists a partition of $n$ into $b+1$ nonnegative integers
\begin{equation}\label{eq_part_n_i}
n = n_1 + n_2 + \cdots + n_b + n^*,
\end{equation}
and there exist nonnegative integers $n^*_1, n^*_2, \ldots, n^*_b
\le n^*$ satisfying
\begin{equation}\label{eq_n*i}
n^*_1 + n^*_2 + \cdots + n^*_b \le \psi_s(b, n^*) + b,
\end{equation}
such that
\begin{equation}\label{eq_recurrence}
\psi_s(m,n) \le 2\psi_{s-1}(m, n^*) + 4m + \sum_{i=1}^b \Bigl(
\psi_{s-2}(m_i, n^*_i) + \psi_s(m_i, n_i) \Bigr).
\end{equation}
\end{recurrence}

Here it is appropriate to repeat Matou\v sek's advice
\cite[p.~179]{matou_book} to first study the proof below and then
try to understand the statement of the recurrence.

\begin{proof}
Let $S$ be a maximum-length Davenport--Schinzel sequence of order
$s$ that is partitionable into $m$ blocks $S_1, \ldots, S_m$, each
of distinct symbols. Thus, $|S| = \psi_s(m,n)$. Group the blocks
$S_1, \ldots, S_m$ into $b$ \emph{layers} $L_1, L_2, \ldots, L_b$
from left to right, by letting each layer $L_i$ contain $m_i$
consecutive blocks.

We partition the alphabet of $S$ into two sets of symbols. The
\emph{local} symbols are those that appear in only one layer, and
the \emph{global} symbols are those that appear in two or more
layers. Let $n_i$ be the number of symbols local to layer $L_i$, for
$1\le i\le b$, and let $n^*$ be the number of global symbols.
Equation (\ref{eq_part_n_i}) follows.

For each layer $L_i$, let $n^*_i$ denote the number of global
symbols that appear in $L_i$. Trivially $n^*_i \le n^*$ for all $i$.
To see that (\ref{eq_n*i}) holds, build a subsequence $S'$ of $S$ by
taking, for each layer $L_i$ and each global symbol $a$ in $L_i$,
just the first occurrence of $a$ within $L_i$. The sequence $S'$,
being a subsequence of $S$, does not contain any alternation of
length $s+2$. Furthermore, $S'$ consists of $b$ blocks of distinct
symbols, corresponding to the $b$ layers of $S$.

However, $S'$ might contain pairs of adjacent equal symbols at the
interface between blocks. But there are at most $b-1$ such pairs of
symbols, and by deleting one symbol from each pair, we finally
obtain a Davenport--Schinzel sequence. Bound (\ref{eq_n*i}) follows.

Each occurrence of a global symbol $a$ in a layer $L_i$ is
classified into \emph{starting}, \emph{middle}, or \emph{ending}, as
follows: If $a$ does not appear in any previous layer $L_j$, $j<i$,
we say that $a$ is a \emph{starting symbol} for $L_i$. Similarly, if
$a$ does not appear in any subsequent layer $L_j$, $j>i$, then $a$
is an \emph{ending symbol} for $L_i$. If $a$ appears both before and
after $L_i$, then $a$ is a \emph{middle symbol} for $L_i$.

Decompose $S$ into four sequences $T_1, T_2, T_3, T_4$ (not
necessarily contiguous), as follows: Let $T_1$ contain all
occurrences of the local symbols of $S$. Let $T_2$ contain all
occurrences of the starting global symbols in all the layers of $S$;
similarly, let $T_3$ contain all occurrences of the middle global
symbols, and let $T_4$ contain all occurrences of the ending global
symbols in all the layers of $S$. Thus, $|T_1| + \cdots + |T_4| =
\psi_s(m,n)$. Each sequence $T_1, \ldots, T_4$ inherits from $S$ the
partition into $b$ layers, in which the $i$-th layer is further
partitioned into $m_i$ blocks.

Each of the sequences $T_1, \ldots, T_4$ might contain pairs of
adjacent equal symbols, but these can only occur at the interface
between adjacent blocks. Hence, by removing at most $m-1$ symbols
from each sequence, we obtain sequences $T'_1, \ldots, T'_4$ with no
adjacent equal symbols. Thus, $\psi_s(m,n) \le |T'_1| + \cdots +
|T'_4| + 4m$. We now bound each of $|T'_1|, \ldots, |T'_4|$
individually.

Let us first consider $T'_1$. The $i$-th layer in $T'_1$ is a
Davenport--Schinzel sequence of order $s$ on $n_i$ symbols, and it
consists of $m_i$ blocks, each of distinct symbols. Thus
\begin{equation*}
|T'_1| \le \sum_{i=1}^b \psi_s(m_i,n_i).
\end{equation*}

Next consider $T'_2$. We claim that each layer in $T'_2$ is a
Davenport--Schinzel sequence of order $s-1$. Indeed, suppose for a
contradiction that some layer in $T'_2$ contains an alternation
$abab\ldots$ of length $s+1$. Then, since $a$ and $b$ are starting
symbols for this layer, they must both appear in $S$ in some
subsequent layer, and so $S$ would contain an alternation of length
$s+2$, a contradiction.

Furthermore, since each global symbol is a starting symbol for
exactly one layer, the layers in $T'_2$ have pairwise disjoint sets
of symbols, so \emph{all} of $T'_2$ is a Davenport--Schinzel
sequence of order $s-1$. A similar argument applies for $T'_4$.
Thus,
\begin{equation*}
|T'_2|, |T'_4| \le \psi_{s-1}(m, n^*).
\end{equation*}

Finally, consider $T'_3$. Each layer in $T'_3$ is composed of middle
global symbols, which appear in $S$ in both previous and subsequent
layers. Therefore, no layer in $T'_3$ can contain an alternation of
length $s$, or else $S$ would contain an alternation of length
$s+2$. Thus, each layer in $T'_3$ is a Davenport--Schinzel sequence
of order $s-2$. (However, the whole $T'_3$ is not necessarily of
order $s-2$.) Since the $i$-th layer in $T'_3$ contains $n^*_i$
different symbols and is partitioned into $m_i$ blocks, each of
distinct symbols, we have
\begin{equation*}
|T'_3| \le \sum_{i=1}^b \psi_{s-2}(m_i, n^*_i).
\end{equation*}

Bound (\ref{eq_recurrence}) follows.
\end{proof}

\begin{remark}\label{remark_upper_improv}
Our key improvement over the method of Agarwal, Sharir, and Shor
lies in the bound for $|T'_3|$. They noted that each layer in $T'_3$
is a sequence of order $s-2$, but they did not use the fact that the
blocks in each layer have distinct symbols. In addition, they did
not introduce the variables $n^*_i$.
\end{remark}

\subsection{Applying the recurrence relation}

We apply Recurrence~\ref{recurrence_psi} repeatedly to obtain
successively better upper bounds on $\psi_s(m,n)$. We first obtain a
polylogarithmic bound, and then we use induction to go all the way
down the inverse Ackermann hierarchy.

For $s\ge 3$ let $m_0(s)$ be a large enough constant (depending only
on $s$) such that
\begin{equation}\label{eq_m0s}
m \ge 2 + 2\lceil \log_2 m\rceil^{s-2} \qquad \text{for all $m\ge
m_0(s)$}.
\end{equation}
Define integers $P_{s,2}$, $Q_{s,2}$ for $s\ge 1$ by
\begin{equation} \label{eq_Ps2_Qs2_start}
P_{1,2} = P_{2,2} = 0, \qquad Q_{1,2} = 1,\ Q_{2,2} = 2,
\end{equation}
and, for $s\ge 3$,
\begin{equation} \label{eq_Ps2_Qs2_general}
\begin{split}
P_{s,2} &= 4P_{s-1,2} + 2P_{s-2,2} + 2Q_{s-2,2} + 8,\\
Q_{s,2} &= \max{\bigl\{ m_0(s),\ 2Q_{s-1,2} + 2Q_{s-2,2} \bigr\}}.
\end{split}
\end{equation}
The reason for our choice of $m_0(s)$ will become apparent later on,
in the proof of Lemma~\ref{lemma_psi_alphak_bound}. (Also recall
that we take $s$ to be a constant, so the growth of $P_{s,2}$,
$Q_{s,2}$ in $s$ is irrelevant for us.)

Our polylogarithmic bound is as follows:

\begin{lemma}\label{lemma_psi_log_bound}
For all $m$, $n$, and $s$, we have
\begin{equation*}
\psi_s(m,n) \le P_{s,2}\, m (\log_2 m)^{s-2} + Q_{s,2}\, n.
\end{equation*}
\end{lemma}

\begin{proof}
We proceed by induction on $s$. If $s=1$ then $\psi_1(m,n) \le n$,
and if $s=2$ then $\psi_2(m,n) \le 2n-1$, and the claim holds. So
let $s\ge 3$.

For each $s$ we proceed by induction on $m$. If $m \le m_0(s)$ then
$\psi_s(m,n) \le m_0(s) n \le Q_{s,2}\, n$, and we are done. So
assume $m > m_0(s)$.

We apply Recurrence~\ref{recurrence_psi} with $b = 2$. Let $m_1 =
\lfloor m/2 \rfloor$ and $m_2 = \lceil m/2 \rceil$, so $m_1 + m_2 =
m$. Let us bound each term in the right-hand side of
(\ref{eq_recurrence}) separately.

The term $2\psi_{s-1}(m,n^*)$ is bounded, by induction on $s$, by
\begin{equation*}
2\psi_{s-1}(m,n^*) \le 2P_{s-1,2}\, m (\log_2 m)^{s-3} +
2Q_{s-1,2}\, n^*.
\end{equation*}
Next, we bound the term $\sum_{i=1}^2 \psi_{s-2}(m_i,n_i^*)$. Using
again induction on $s$, and applying $\log_2 m_i \le \log_2 m$, we
get
\begin{equation*}
\sum_{i=1}^2 \psi_{s-2}(m_i,n_i^*) \le P_{s-2,2}\, m (\log_2
m)^{s-4} + Q_{s-2,2}(n^*_1 + n^*_2).
\end{equation*}
Now, applying (\ref{eq_n*i}), we bound $n^*_1 + n^*_2$ loosely by
$n^*_1 + n^*_2 \le \psi_s(2,n^*) + 2 \le 2n^* + m$. Thus, being
again very loose, we get
\begin{equation*}
\sum_{i=1}^2 \psi_{s-2}(m_i,n_i^*) \le m (\log_2 m)^{s-3} (P_{s-2,2}
+ Q_{s-2,2}) + 2Q_{s-2,2}\, n^*.
\end{equation*}
Next we bound the term $\sum_{i=1}^2 \psi_s(m_i,n_i)$, using
induction on $m$. Applying $\log_2 m_i \le \log_2 m - {1\over 2}$,
which is true for $m\ge 3$, and using the fact that $(x - {1\over
2})^{s-2} \le x^{s-2} - {1\over 2}x^{s-3}$ for all $x\ge {1\over
2}$, we get
\begin{align*}
\sum_{i=1}^2 \psi_s(m_i,n_i) &\le \sum_{i=1}^2 \left( P_{s,2}\, m_i
(\log_2
m_i)^{s-2} + Q_{s,2}\, n_i \right)\\
&\le P_{s,2}\, m (\log_2 m)^{s-2} - {1\over 2} P_{s,2}\, m(\log_2
m)^{s-3} + Q_{s,2} (n - n^*).
\end{align*}
Finally, we bound $4m$ (very loosely for $s\ge 4$) by $4m(\log_2
m)^{s-3}$. Putting everything together, we get
\begin{align*}
\psi_s(m,n) &\le P_{s,2}\, m (\log_2 m)^{s-2} + Q_{s,2}\, n \\
& \qquad {}+ m (\log_2 m)^{s-3}\left( 2P_{s-1,2} + P_{s-2,2} +
Q_{s-2,2} +
4 - {1\over 2} P_{s,2} \right)\\
& \qquad {} + (2Q_{s-1,2} + 2Q_{s-2,2} - Q_{s,2})n^*.
\end{align*}
By the definition of $P_{s,2}$ and $Q_{s,2}$ in
(\ref{eq_Ps2_Qs2_general}), the last two lines are non-positive, so
\begin{equation*}
\psi_s(m,n) \le P_{s,2}\, m (\log_2 m)^{s-2} + Q_{s,2}\, n. \qedhere
\end{equation*}
\end{proof}

We are now ready to go all the way down the inverse Ackermann
hierarchy. Define integers $P_{s,k}$, $Q_{s,k}$ for $k\ge 3$, $s\ge
1$ by
\begin{equation*}
P_{1,k} = P_{2,k} = 0, \qquad Q_{1,k} = 1,\ Q_{2,k} = 2,
\end{equation*}
and, for $s\ge 3$,
\begin{equation}\label{eq_Psk_Qsk}
\begin{split}
P_{s,k} &= Q_{s-2,k}(1 + P_{s,k-1}) + 2d_s\, P_{s-1,k} + d'_s\,
P_{s-2,k} + 4,\\
Q_{s,k} &= Q_{s-2,k}\, Q_{s,k-1} + 2Q_{s-1,k},
\end{split}
\end{equation}
for some constants $d_s$ and $d'_s$ to be specified later, with
$P_{s,2}$, $Q_{s,2}$ as in (\ref{eq_Ps2_Qs2_start}),
(\ref{eq_Ps2_Qs2_general}). These quantities $P_{s,k}$, $Q_{s,k}$ will give rise to $C_{s,k}$ of Lemma~\ref{lemma_gral_psi_alphak_bound}.

\begin{lemma}\label{lemma_psi_alphak_bound}
For every $s$ there exists a constant $c_s$ such that
\begin{equation}\label{eq_psi_alphak_bound}
\psi_s(m,n) \le P_{s,k}\, m (\alpha_k(m)+c_s)^{s-2} + Q_{s,k}\, n
\end{equation}
for all integers $n$, $m$, $s$, and $k$.
\end{lemma}

The proof is similar to the proof of
Lemma~\ref{lemma_psi_log_bound}, though more complex, since we
proceed by induction on $k$ for each $s$. Before delving into the
actual details, we give a brief sketch of the proof. For the
purposes of this sketch, denote the right-hand side of
(\ref{eq_psi_alphak_bound}) by $\Gamma_{s,k}(m,n)$. Now refer to
equation (\ref{eq_recurrence}) in Recurrence~\ref{recurrence_psi}.

The proof proceeds as follows. We bound the term $\psi_{s-1}(m,n^*)$
by $\Gamma_{s-1,k}(m,n^*)$. We bound the terms
$\psi_{s-2}(m_i,n^*_i)$ by $\Gamma_{s-2,k}(m_i,n^*_i)$; this
produces the term $Q_{s-2,k}\sum n^*_i$, on which we apply
(\ref{eq_n*i}). We bound the resulting term $\psi_s(b,n^*)$ by
$\Gamma_{s,k-1}(b,n^*)$ (here is where we use induction on $k$).
Finally, we bound the terms $\psi_s(m_i,n_i)$ by
$\Gamma_{s,k}(m_i,n_i)$ by induction on $m$ (since $m_i < m$ for
every $i$).

\begin{proof}[Proof of Lemma~\ref{lemma_psi_alphak_bound}]
By induction on $s$. As before, the claim is easily established for
$s = 1,2$, so assume $s\ge 3$ is fixed.

For each $s$ we proceed by induction on $k$. If $k=2$ then the claim
reduces to Lemma~\ref{lemma_psi_log_bound}, so assume $k \ge 3$.

By our induction assumption on $s$, we have
\begin{equation}\label{eq_ind_s-1_s-2_alpha}
\begin{split}
\psi_{s-1}(m,n) &\le P_{s-1,k}\, m(\alpha_k(m) + c_{s-1})^{s-3} +
Q_{s-1,k}\, n,\\
\psi_{s-2}(m,n) &\le P_{s-2,k}\, m(\alpha_k(m) + c_{s-2})^{s-4} +
Q_{s-2,k}\, n,
\end{split}
\end{equation}
for all $m$ and $n$.

Here it is convenient to work with a slight variant $\widehat
\alpha_k(x)$ of the inverse Ackermann hierarchy. Define $\widehat
\alpha_k(x)$ for $k\ge 2$, $x\ge 0$ by $\widehat \alpha_2(x) =
\alpha_2(x) = \lceil \log_2 x\rceil$, and for $k\ge 3$ by the
recurrence
\begin{equation}\label{eq_widehat_alpha_k}
\widehat\alpha_k(x) =
\begin{cases}
1, & \text{if $x\le m_0(s)$};\\
1 + \widehat\alpha_k{\left(1 + 2\widehat\alpha_{k-1}(x)^{s-2}
\right)}, & \text{otherwise};
\end{cases}
\end{equation}
with $m_0(s)$ as given in (\ref{eq_m0s}). (Compare
(\ref{eq_widehat_alpha_k}) to the definition (\ref{eq_rec_alpha_k})
of $\alpha_k(x)$; our choice of $m_0(s)$ guarantees that $\widehat
\alpha_k(x)$ is well-defined for all $k$ and $x$.)

The functions $\widehat\alpha_k(x)$ are almost equivalent to the
usual inverse Ackermann functions $\alpha_k(x)$. In fact, there
exists a constant $c_s$, depending only on $s$, such that
$|\widehat\alpha_k(x) - \alpha_k(x)| \le c_s$ for all $k$ and $x$.
See Appendix~B of \cite{interval_chains} for a general technique for
proving bounds of this type (or see Appendix~\ref{app_ack_like} in
this paper).

We will show that
\begin{equation}\label{eq_psi_widehat_alphak}
\psi_s(m,n) \le P_{s,k}\, m\widehat \alpha_k(m)^{s-2} + Q_{s,k}\, n
\end{equation}
for all $n$, $m$, and $k$. We will do this by induction on $k$, and
for each $k$ by induction on $m$. Then our claim will follow.

If $m\le m_0(s)$, then $\psi_s(m,n) \le m_0(s) n \le Q_{s,2}\, n \le
Q_{s,k}\, n$, and we are done. So assume $m > m_0(s)$.

We want to translate the bounds (\ref{eq_ind_s-1_s-2_alpha}) into
bounds involving $\widehat\alpha_k$. Since $\alpha_k(m) \le
\widehat\alpha_k(m) + c_s$ and $\widehat\alpha_k(m) \ge 1$, it
follows (being somewhat slack) that there exist multiplicative
constants $d_s$, $d'_s$ such that
\begin{align}
\psi_{s-1}(m,n) &\le  d_s\, P_{s-1,k}\, m\widehat\alpha_k(m)^{s-3} +
Q_{s-1,k}\, n,\label{eq_ind_s-1}\\
\psi_{s-2}(m,n) &\le  d'_s\, P_{s-2,k}\, m\widehat\alpha_k(m)^{s-4}
+ Q_{s-2,k}\, n,\label{eq_ind_s-2}
\end{align}
for all $n$ and $m$.

Assume by induction on $k$ that (\ref{eq_psi_widehat_alphak}) holds
for $k-1$. Choose
\begin{equation}\label{eq_b}
b = \left\lfloor {m \over \widehat\alpha_{k-1}(m)^{s-2} }
\right\rfloor.
\end{equation}
Let $m_i = \lfloor m/b \rfloor$ or $\lceil m/b \rceil$ for each $i$,
such that $\sum m_i = m$. We claim that
\begin{equation}\label{eq_bd_mi}
m_i \le 1 + 2 \widehat \alpha_{k-1}(m)^{s-2}, \qquad \text{for all }
1\le i\le b.
\end{equation}
Indeed, by our choice of $m_0(s)$ as given in (\ref{eq_m0s}), we
have $\widehat \alpha_{k-1}(m)^{s-2} \le \lceil \log_2 m\rceil
^{s-2} \le m/2$ for all $m\ge m_0(s)$. Thus,
\begin{align*}
m_i \le 1 + {m\over b } &\le 1 + {m \over m/ \widehat
\alpha_{k-1}(m)^{s-2} - 1} \\&= 1 + {m \widehat
\alpha_{k-1}(m)^{s-2} \over m  - \widehat \alpha_{k-1}(m)^{s-2}} \le
1 + {m \widehat \alpha_{k-1}(m)^{s-2} \over m - m/2} = 1 + 2
\widehat \alpha_{k-1}(m)^{s-2}.
\end{align*}

Let us bound each term in the right-hand side of
(\ref{eq_recurrence}). We first bound the term $2\psi_{s-1}(m,n^*)$
using (\ref{eq_ind_s-1}), and we obtain
\begin{equation*}
2\psi_{s-1}(m,n^*) \le 2d_s \, P_{s-1,k}\, m \widehat
\alpha_k(m)^{s-3} + 2 Q_{s-1,k}\, n^*.
\end{equation*}

Next we bound $\sum_{i=1}^b \psi_{s-2}(m_i,n^*_i)$ using
(\ref{eq_ind_s-2}). Observing that $\widehat\alpha_k(m_i) \le
\widehat \alpha_k(m)$,
\begin{align}
\sum_{i=1}^b \psi_{s-2}(m_i,n^*_i) &\le \sum_{i=1}^b \left( d'_s\,
P_{s-2,k}\, m_i \widehat \alpha_k(m_i)^{s-4} + Q_{s-2,k}\, n^*_i
\right) \nonumber \\
&\le d'_s \, P_{s-2,k} \, m \widehat \alpha_k(m)^{s-4} + Q_{s-2,k}
\sum_{i=1}^b n^*_i. \label{eq_sum_psi_mi_n*i}
\end{align}
Now we apply (\ref{eq_n*i}), and we bound $\psi_s(b,n^*)$ by
(\ref{eq_psi_widehat_alphak}) with $k-1$ in place of $k$.
\begin{equation*}
\sum_{i=1}^b n^*_i \le \psi_s(b,n^*) + b \le P_{s,k-1}\, b \widehat
\alpha_{k-1}(b)^{s-2} + Q_{s,k-1}\, n^* + b.
\end{equation*}
By our choice of $b$ in (\ref{eq_b}), we have $\widehat
\alpha_{k-1}(b)^{s-2} \le \widehat \alpha_{k-1}(m)^{s-2} \le m/b$,
so, being somewhat slack,
\begin{equation*}
\sum_{i=1}^b n^*_i \le P_{s,k-1}\, m + Q_{s,k-1}\, n^* + m \le m
\widehat \alpha_k(m)^{s-3} (1 + P_{s,k-1}) + Q_{s,k-1}\, n^*.
\end{equation*}
Substituting this into (\ref{eq_sum_psi_mi_n*i}), and being slack
again, we get
\begin{equation*}
\sum_{i=1}^b \psi_{s-2}(m_i,n^*_i) \le m \widehat \alpha_k(m)^{s-3}
\left( d'_s\, P_{s-2,k} + Q_{s-2,k} (1 + P_{s,k-1}) \right) +
Q_{s-2,k}\, Q_{s,k-1}\, n^*.
\end{equation*}

Next we bound $\sum_{i=1}^b \psi_s(m_i,n_i)$, applying
(\ref{eq_psi_widehat_alphak}) by induction on $m$ (since $m_i < m$):
\begin{equation*}
\sum_{i=1}^b \psi_s(m_i,n_i) \le \sum_{i=1}^b \left( P_{s,k}\, m_i
\widehat\alpha_k(m_i)^{s-2} + Q_{s,k}\, n_i \right).
\end{equation*}
But by (\ref{eq_bd_mi}) and (\ref{eq_widehat_alpha_k}),
\begin{equation*}
\widehat\alpha_k(m_i) \le \widehat\alpha_k{\left( 1 +
2\widehat\alpha_{k-1}(m)^{s-2} \right)} = \widehat\alpha_k(m) -1.
\end{equation*}
Further, we have $(x-1)^{s-2} \le x^{s-2} - x^{s-3}$ for all $x\ge
1$. Therefore,
\begin{equation*}
\sum_{i=1}^b \psi_s(m_i,n_i) \le P_{s,k} \, m
\left(\widehat\alpha_k(m)^{s-2} - \widehat\alpha_k(m)^{s-3}\right) +
Q_{s,k}(n - n^*).
\end{equation*}

Finally, we bound $4m$ very loosely by
$4m\widehat\alpha_k(m)^{s-3}$. Putting everything together, we get
\begin{align*}
\psi_s(m,n) &\le P_{s,k}\, m \widehat\alpha_k(m)^{s-2} + Q_{s,k}\, n\\
& \qquad {}+ m\widehat\alpha_k(m)^{s-3} \left( 2d_s\, P_{s-1,k} +
d'_s\, P_{s-2,k} + Q_{s-2,k}(1 + P_{s,k-1}) + 4 -
P_{s,k}\right) \\
& \qquad {} + (2Q_{s-1,k} + Q_{s-2,k}\, Q_{s,k-1} - Q_{s,k}) n^*.
\end{align*}
By the definition of $P_{s,k}$ and $Q_{s,k}$ in (\ref{eq_Psk_Qsk}),
the last two lines equal zero, and we get
\begin{equation*}
\psi_s(m,n) \le P_{s,k}\, m\widehat \alpha_k(m)^{s-2} + Q_{s,k}\, n.
\qedhere
\end{equation*}
\end{proof}

All that remains is to analyze the asymptotic growth of $P_{s,k}$,
$Q_{s,k}$ in $k$ for fixed $s$. We have
\begin{equation*}
P_{3,k}, Q_{3,k} = \Theta(k), \qquad P_{4,k}, Q_{4,k} =
\Theta\bigl(2^k \bigr),
\end{equation*}
and, in general, letting $t = \lfloor (s-2)/2 \rfloor$,
\begin{equation}\label{eq_bd_PQ}
P_{s,k}, Q_{s,k} =
\begin{cases}
2^{(1/t!) k ^t \pm O{\left( k^{t-1} \right)}}, & \text{$s\ge 4$ even};\\
2^{(1/t!) k ^t \log_2 k \pm O{\left( k^t \right)}}, & \text{$s\ge 3$
odd}
\end{cases}
\end{equation}
(see Appendix~\ref{app_growth_constants} for the proof). Thus,
Lemma~\ref{lemma_psi_alphak_bound} is equivalent to
Lemma~\ref{lemma_gral_psi_alphak_bound}.

\begin{remark}
The investment we made in using a more complicated recurrence
(Recurrence~\ref{recurrence_psi} instead of the one used by Agarwal
et al.~\cite{ASS89, DS_book}) paid off in
Lemma~\ref{lemma_psi_alphak_bound}. Besides being tighter, the lemma
also has a simpler form. The corresponding bound in
\cite{ASS89,DS_book} is of the form
\begin{equation*}
\psi_s(m,n) \le \mathcal F_{s,k}(n) \cdot m \alpha_k(m) + \mathcal
G_{s,k}(n) \cdot n,
\end{equation*}
where $\mathcal F_{s,k}(n)$ and $\mathcal G_{s,k}(n)$ are functions
of $\alpha(n)$. Our constants $P_{s,k}$, $Q_{s,k}$, in contrast, do
not depend on $n$.
\end{remark}

\section{A new technique for bounding $\psi_s(m,n)$}\label{sec_psi_new}

We now present an alternative technique for bounding $\psi_s(m,n)$.
Our new technique is based on a variant of Davenport--Schinzel
sequences, in which we turn the problem around, in a sense. We call
our variant sequences \emph{almost-DS sequences}.

An \emph{almost-DS sequence of order $s$ with multiplicity $k$ and
$m$ blocks} (or an $\ADS^s_k(m)$-sequence, for short) is a sequence
that satisfies the following properties:
\begin{itemize}
\item
It is a concatenation of $m$ blocks, each block containing only
distinct symbols.

\item
Each symbol appears at least $k$ times (in different blocks, so we
must have $m\ge k$ for there to be any symbols at all).

\item
The sequence contains no alternation $abab\ldots$ of length $s+2$.
\end{itemize}

Note that we do allow repetitions at the interface between adjacent
blocks (this simplifies matters). This is why these are
\emph{almost} Davenport--Schinzel sequences.

We now pose a different problem: We ask for \emph{maximizing the
number of distinct symbols}. Let $\N^s_k(m)$ denote the maximum
number of distinct symbols in an $\ADS^s_k(m)$-sequence. (Note that
$\N^s_k(m)=0$ for $m<k$.)

The connection between $\psi_s(m,n)$ and $\N^s_k(m)$ is based on the
following lemma:

\begin{lemma}\label{lemma_psi_to_ADS}
For all $s$, $n$, $m$, and $k$ we have $\psi_s(m,n) \le k \bigl(
\N^s_k(m) + n \bigr)$.
\end{lemma}

\begin{proof}
Let $S$ be a maximum-length Davenport--Schinzel sequence of order
$s$ on $n$ distinct symbols that is partitionable into $m$ blocks,
each of distinct symbols. Thus, $|S| = \psi_s(m,n)$. Let $k\ge 1$ be
a parameter.

We transform $S$ into another sequence $S'$ in which every symbol
appears exactly $k$ times as follows:\footnote{A similar argument
has been used by Sundar \cite[Lemma~9]{sundar} for a different
problem.} For each symbol $a$, group the occurrences of $a$ in $S$
from left to right into ``clusters" of size $k$, deleting the last
remaining $\le k-1$ occurrences of $a$. Make the occurrences of $a$
in different clusters different, by replacing each $a$ in the $i$-th
cluster by a new symbol $a_i$.

We deleted at most $kn$ symbols from $S$, so $|S'| \ge |S| - kn$. On
the other hand, $S'$ is clearly an $\ADS^s_k(m)$-sequence (the
symbol deletions might have created repetitions at the interface
between blocks, but these are permitted in almost-DS sequences; on
the other hand, the symbol replacements do not introduce any
forbidden alternations). Thus, $S'$ contains at most $\N^s_k(m)$
distinct symbols. Since each symbol appears exactly $k$ times, we
have $|S'| \le k \cdot \N^s_k(m)$. The claim follows.
\end{proof}

Thus, the problem of bounding $\psi_s(m,n)$ reduces to bounding
$\N^s_k(m)$.

\subsection{Bounding the number of symbols in ADS sequences}

We first derive some basic results: For every $s\ge 1$, if we take
$k\le s$ then $\N^s_k(m) = \infty$, but if we take $k = s+1$ then
$\N^s_k(m)$ is already finite.

\begin{lemma}\label{lemma_nss_infty}
For all $s \ge 1$, $m\ge s$ we have $\N^s_s(m) = \infty$.
\end{lemma}

\begin{proof}
Take the sequence
\begin{equation*}
abc\ldots\ \ldots cba\ abc\ldots\ \ldots
\end{equation*}
with $s$ blocks, with arbitrarily many symbols in each block. Each
symbol appears $s$ times, and the maximum alternation is of length
$s+1$.
\end{proof}

\begin{lemma}\label{lemma_n12}
We have $\N^1_2(m) = m - 1$.
\end{lemma}

\begin{proof}
Let $S$ be an $\ADS^1_2(m)$-sequence. Since $S$ cannot contain an
alternation $aba$, each symbol must have all its occurrences
contiguous. Given that $S$ contains $m$ blocks, the sequence that
maximizes the number of distinct symbols is
\begin{equation*}
1\ 12\ 23\ \ldots\ (m-2)(m-1)\ (m-1),
\end{equation*}
with $m-1$ distinct symbols.
\end{proof}

\begin{lemma}\label{lemma_ns_s1}
For all $s\ge 2$ we have $\N^s_{s+1}(m) \le {m-2\choose s-1} =
O{\left( m^{s-1} \right)}$.
\end{lemma}

\begin{proof}
Suppose for a contradiction that there exists an
$\ADS^s_{s+1}(m)$-sequence $S$ with $n = 1 + {m-2\choose s-1}$
distinct symbols. Thus, each symbol appears in at least $s+1$ out of
$m$ different blocks. For each symbol $a$, consider the $s-1$
``internal" occurrences of $a$, meaning, all occurrences except the
first and the last. These internal occurrences can fall in any of
the $m-2$ ``internal" blocks of $S$ (excluding the first and last
blocks).

By our choice of $n$, there must be two symbols $a$, $b$ whose
internal occurrences fall in the same $s-1$ out of $m-2$ internal
blocks. These internal occurrences create an alternation of length
at least $s$ (in the best case, they form the subsequence $ab\ ba\
ab\ \ldots$). Since both $a$ and $b$ also appear before and after
this subsequence, $S$ contains an alternation of length $s+2$, a
contradiction.
\end{proof}

We now bound $\N^s_k(m)$ by deriving recurrences and solving them,
in a manner almost entirely analogous to \cite{interval_chains}. We
begin with the following recurrence and corollary, which are
analogous to Lemma 3.2 in \cite{interval_chains}:

\begin{recurrence}\label{rec_log}
For every $s\ge 3$ and every $k$ and $m$ we have
\begin{equation*}
\N^s_{2k-1}(2m) \le 2 \N^s_{2k-1}(m) + 2 \N^{s-1}_k(m).
\end{equation*}
\end{recurrence}

\begin{proof}
Given an $\ADS^s_{2k-1}(2m)$-sequence $S$, partition the $2m$ blocks
of $S$ into a ``left half" and a ``right half" of $m$ blocks each.
The symbols of $S$ fall into four categories:
\begin{itemize}
\item
Symbols that appear only in the left half. Taking just these symbols
produces an $\ADS^s_{2k-1}(m)$-sequence, so there are at most
$\N^s_{2k-1}(m)$ such symbols.

\item
Symbols that appear only in the right half. There are also at most
$\N^s_{2k-1}(m)$ such symbols.

\item
Symbols that appear in both halves, but appear at least $k$ times in
the left half. Taking just these symbols, and just their left-half
occurrences, produces an $\ADS^{s-1}_k(m)$-sequence $S'$. (An
alternation $abab\ldots$ of length $s+1$ in $S'$ would be extended
to length $s+2$ by an $a$ or $b$ that appears in the right half.)
Thus, there are at most $\N^{s-1}_k(m)$ of these symbols.

\item
Symbols that appear in both halves, but appear at least $k$ times in
the right half. There are also at most $\N^{s-1}_k(m)$ such symbols.
\qedhere
\end{itemize}
\end{proof}

\begin{corollary}\label{cor_log}
For every fixed $s\ge 2$, if we let $k = 2^{s-1} + 1$, then
\begin{equation*}
\N^s_k(m) = O{\left( m (\log m)^{s-2} \right)}
\end{equation*}
(where the constant implicit in the $O$ notation might depend on
$s$).
\end{corollary}

\begin{proof}
Apply Recurrence~\ref{rec_log} using induction on $s$, using
Lemma~\ref{lemma_ns_s1} as base case for $s = 2$.
\end{proof}

The following recurrence and corollary for $\N^3_k(m)$ are analogous
to Recurrence~3.3 and Lemma~3.5 in \cite{interval_chains}:

\begin{recurrence}\label{rec_N3km}
Let $t$ be an integer parameter, with $t\le \sqrt m$. Then,
\begin{equation*}
\N^3_k(m) \le \left( 1+ {m\over t} \right) \N^3_k(t) +
\N^3_{k-2}{\left( 1 + {m\over t} \right)} + 3m.
\end{equation*}
\end{recurrence}

\begin{proof}
Take a sequence $S$ that maximizes $\N^3_k(m)$. Let $b = \lceil m/t
\rceil \le 1 + m/t$. Partition the $m$ blocks of $S$ from left to
right into $b$ \emph{layers} $L_1,\ldots, L_b$ of at most $t$ blocks
each.

We classify the symbols of $S$ into different types. A symbol is
\emph{local} for layer $L_i$ if it only appears in $L_i$. Taking
just the symbols local to $L_i$ produces an $\ADS^3_k(t)$-sequence.
Therefore, the number of local symbols is at most $\N^3_k(t)$ per
layer, or at most $b\N^3_k(t) \le \left( 1 + {m\over t} \right)
\N^3_k(t)$ altogether.

Symbols which appear in at least two layers are called \emph{global
symbols}.

Call a global symbol \emph{left-concentrated} for layer $L_i$ if it
makes its first appearance in $L_i$, and it appears at least three
times in $L_i$. Given a layer $L_i$, take just the left-concentrated
symbols for $L_i$, and just their occurrences within $L_i$. The
resulting sequence $S'_i$ cannot contain an alternation $abab$, or
else $S$ would contain $ababa$. Therefore, $S'_i$ is an
$\ADS^2_3(t)$-sequence, so by Lemma~\ref{lemma_ns_s1} it has at most
$t-2$ different symbols. Thus, there are at most $b(t-2)\le \bigl(1
+ {m\over t} \bigr) (t-2) \le m$ left-concentrated symbols
altogether (since $t\le \sqrt m$).

Similarly, there are at most $m$ \emph{right-concentrated} symbols.

Next, call a global symbol \emph{middle-concentrated} for layer
$L_i$ if it appears at least twice in $L_i$, and it also appears
before $L_i$ and after $L_i$.

Given $L_i$, take just the middle-concentrated symbols for $L_i$,
and just their occurrences within $L_i$. The resulting sequence
$S''_i$ cannot contain an alternation $aba$, so $S''_i$ is an
$\ADS^1_2(t)$-sequence, and so by Lemma~\ref{lemma_n12} it contains
at most $t-1$ different symbols. Therefore, there are at most
$b(t-1) \le m$ middle-concentrated symbols. (Note that we might have
counted the same middle-concentrated symbol more than once.)

Finally, take all the global symbols we have not accounted for so
far---the \emph{scattered symbols}. Each of these symbols must
appear in at least $k-2$ different layers. Build a subsequence of
$S$ by taking just the scattered symbols, and for each scattered
symbol, just one occurrence per layer. Each layer becomes a block,
and no new forbidden alternation can arise. Hence, we get an
$\ADS^3_{k-2}(b)$-sequence, which can have at most
$\N^3_{k-2}{\left( 1 + {m\over t}\right)}$ different symbols.
\end{proof}

\begin{corollary}\label{cor_N3km_upper}
There exists an absolute constant $c$ such that, for every $k\ge 2$,
we have
\begin{equation*}
\N^3_{2k+1}(m) \le cm \alpha_k(m) \qquad \text{for all $m$}.
\end{equation*}
\end{corollary}

\begin{proof}
Let $m_0$ be a constant large enough that
\begin{equation*}
m\ge 1+ 9\lceil \log_2 m\rceil^2 \qquad \text{for all $m\ge m_0$}.
\end{equation*}
We will work with a slight variant of the inverse Ackermann
function. For this proof, let $\widehat \alpha_k(x)$, $k\ge 2$, be
given by $\widehat \alpha_2(x) = \alpha_2(x) = \lceil \log_2
x\rceil$, and, for $k\ge 3$, by the recurrence
\begin{equation*}
\widehat\alpha_k(x) =
\begin{cases}
1, & \text{if $x\le m_0$};\\
1+ \widehat\alpha_k{\bigl( 3 \widehat\alpha_{k-1}(x) \bigr)}, &
\text{otherwise}.
\end{cases}
\end{equation*}
Note that $\widehat\alpha_k(x)$ is well-defined by our choice of
$m_0$. Furthermore, there exists a constant $c_0$ such that
$|\widehat\alpha_k(x) - \alpha_k(x)| \le c_0$ for all $k$ and $x$
(see Appendix~B of~\cite{interval_chains}, or
Appendix~\ref{app_ack_like} in this paper).

We will prove by induction on $k\ge 2$ that
\begin{equation*}
\N^3_{2k+1}(m) \le c_1 m \widehat \alpha_k(m) \qquad \text{for all
$m$},
\end{equation*}
for some absolute constant $c_1$; this implies our claim. By
Corollary~\ref{cor_log} we have $\N^3_5(m) = O(m\log m)$, so the
base case $k=2$ follows by choosing $c_1$ sufficiently large. We
choose $c_1$ large enough so that it also satisfies $c_1 \ge 12$ and
that $c_1 \ge \N^3_7(m) / m$ for all $m\le m_0$.

Now, let $k\ge 3$, and assume the bound holds for $k-1$. To
establish the bound for $k$, first let $m\le m_0$. Then we have
\begin{equation*}
\N^3_{2k+1}(m) \le \N^3_7(m) \le c_1 m = c_1 m \widehat\alpha_k(m),
\end{equation*}
since $\N^3_k(m)$ is nonincreasing in $k$ for fixed $m$. Thus, let
$m
> m_0$. We apply Recurrence~\ref{rec_N3km} with $t = 3 \widehat
\alpha_{k-1}(m)$. (Note that $t\le \sqrt m$ for $m>m_0$ by our
choice of $m_0$.) By the induction assumption for $k-1$ we have
\begin{equation*}
\N^3_{2k-1}{\left( 1 + {m\over t} \right)} \le \N^3_{2k-1}{\left(
{2m\over t} \right)} \le {2c_1m\over t} \widehat\alpha_{k-1}{\left(
{2m\over t} \right)} \le {2c_1m\over t} \widehat\alpha_{k-1}(m) =
{2c_1m \over 3}.
\end{equation*}
Substituting into Recurrence~\ref{rec_N3km}, and letting
$\N^3_{2k+1}(m) = mg(m)$,
\begin{align*}
g(m) &\le g(t) + {\N^3_{2k+1}(t)\over m} + {2c_1\over 3} + 3\\
&\le g(t) + {2c_1\over 3} + 4 &&\text{(since, by
Lemma~\ref{lemma_ns_s1}, $\N^3_{2k+1}(t) \le t^2
\le m$)}\\
&\le g(t) + c_1 && \text{(since $c_1\ge 12$)}.
\end{align*}
Since $\widehat \alpha_k(t) = \widehat \alpha_k(m) - 1$, it follows
by induction on $m$ (with base case $m\le m_0$) that
\begin{equation*}
g(m) \le c_1 \widehat \alpha_k(m) \qquad \text{for all $m$}.
\end{equation*}
Therefore,
\begin{equation*}
\N^3_{2k+1}(m) \le c_1 m \widehat\alpha_k(m) \qquad \text{for all
$m$}. \qedhere
\end{equation*}
\end{proof}

The bound for $\psi_3(m,n)$ in
Lemma~\ref{lemma_gral_psi_alphak_bound} now follows from
Corollary~\ref{cor_N3km_upper} and Lemma~\ref{lemma_psi_to_ADS}.

\subsection{Obtaining Klazar's improved upper bound for $\lambda_3(n)$}

Klazar's tighter upper bound (\ref{eq_lambda_3_old_sandwich}) for
$\lambda_3(n)$ follows by using the following relation between
$\lambda_3(n)$ and $\psi_3(m,n)$, instead of
Lemma~\ref{lemma_lambda_to_psi}:

\begin{lemma}[Klazar~\cite{klazar99}]\label{lemma_lamda_3_to_psi_klazar}
We have $\lambda_3(n) \le \psi_3(1 + 2n/\ell, n) + 3n\ell$, where
$\ell\le n$ is a free parameter.
\end{lemma}

(Klazar actually proved this relation under a stricter definition of
$\psi_3(m,n)$.) For completeness, we prove
Lemma~\ref{lemma_lamda_3_to_psi_klazar} in
Appendix~\ref{app_lambda_3_klazar}.

\begin{corollary}
$\lambda_3(n) \le 2n\alpha(n) + O{\bigl( n\sqrt{\alpha(n)} \bigr)}$.
\end{corollary}

\begin{proof}
Taking $s=3$ and $k = 2\alpha(m) + 1$ in
Lemma~\ref{lemma_psi_to_ADS}, and bounding $\N^3_{2\alpha(m) +
1}(m)$ by Corollary~\ref{cor_N3km_upper}, we get
\begin{equation*}
\psi_3(m,n) \le \bigl(2\alpha(m) + 1\bigr)
\bigl(cm\alpha_{\alpha(m)}(m) + n\bigr) = 2n\alpha(m) + n+
O{\bigl(m\alpha(m) \bigr)}.
\end{equation*}
We now apply Lemma~\ref{lemma_lamda_3_to_psi_klazar} with $\ell =
\sqrt{\alpha(n)}$.
\end{proof}

\subsection{Bounding $\N^s_k(m)$ for general $s$}

The following recurrence and corollary for $\N^s_k(m)$ are analogous
to Recurrence~3.6 and Lemma~3.8 in \cite{interval_chains}:

\begin{recurrence}\label{rec_nskm}
Let $s\ge 3$ be fixed. Let $k_1$, $k_2$, $k_3$ be integers, and put
$k = k_2 k_3 + 2k_1 - 3k_2 - k_3 + 2$. Then,
\begin{equation*}
\N^s_k(m) \le \left(1+ {m\over t}\right) \left( \N^s_k(t) +
2\N^{s-1}_{k_1}(t) + \N^{s-2}_{k_2}(t) \right) + \N^s_{k_3}{\left(
1+{m\over t} \right)},
\end{equation*}
where $t$ is a free parameter.
\end{recurrence}

\begin{proof}
Take a sequence $S$ that maximizes $\N^s_k(m)$. Again partition the
$m$ blocks of $S$ into $b = \lceil m/t \rceil \le 1 + m/t$ layers
$L_1,\ldots, L_b$, with at most $t$ blocks per layer.

We again classify the symbols of $S$ into \emph{local} (if the
symbol appears in only one layer), or \emph{global}. As before,
there are at most $\left( 1 + {m\over t} \right) \N^s_k(t)$ local
symbols.

And we again classify the global symbols into
\emph{left-concentrated}, \emph{right-concentrated},
\emph{middle-concentrated}, and \emph{scattered}. This time we do
this as follows:

A global symbol is \emph{left-concentrated} for layer $L_i$ if its
first $k_1$ occurrences fall in $L_i$. The overall number of
left-concentrated symbols is at most $\left( 1 + {m\over t} \right)
\N^{s-1}_{k_1}(t)$. \emph{Right-concentrated} symbols are defined
and handled analogously.

A global symbol is \emph{middle-concentrated} for layer $L_i$ if it
appears at least $k_2$ times in $L_i$, and it also appears before
$L_i$ and after $L_i$. There are at most $\left( 1 + {m\over t}
\right) \N^{s-2}_{k_2}(t)$ middle-concentrated symbols altogether.

Finally, a global symbol is \emph{scattered} if it appears in at
least $k_3$ different layers. Taking just these symbols, and for
each symbol, just one occurrence per layer, we obtain an
$\ADS^s_{k_3}{\left( b \right)}$-sequence. Thus, there are at most
$\N^s_{k_3}(b) \le \N^s_{k_3}{\left(1+ {m\over t}\right)}$ scattered
symbols.

All that remains is to show that we did not miss any global symbol.
Suppose a global symbol is neither \hbox{left-,} middle-, nor
right-concentrated, nor scattered. Then the symbol appears at most
$2(k_1 - 1) + (k_3 - 3)(k_2 - 1) = k - 1$ times in $S$, a
contradiction.
\end{proof}

The only significant difference between Recurrence \ref{rec_nskm}
above and Recurrence 3.6 in \cite{interval_chains} lies in the
formula for $k$ in terms of $k_1$, $k_2$, and $k_3$. (The formula
there is $k = k_2 k_3 + 2k_1 - 2k_2$.) But in both cases we get the
same asymptotic behavior:

\begin{corollary}\label{cor_nskm}
Define $R_s(d)$ for $s\ge1$, $d\ge2$ by $R_1(d) = 2$, $R_2(d) = 3$,
and for $s\ge 3$ by
\begin{align*}
R_s(2) &= 2^{s-1} + 1,\\
R_s(d) &= R_s(d-1) R_{s-2}(d) + 2 R_{s-1}(d) - 3R_{s-2}(d)\\
&\qquad - R_s(d-1) + 2, \qquad \text{for $d\ge 3$}.
\end{align*}
Then, for every $s\ge 2$ and $d\ge 2$, if $k\ge R_s(d)$ then
\begin{equation*}
\N^{s}_k(m) \le c m \alpha_d(m)^{s-2} \qquad \text{for all $m$}.
\end{equation*}
Here $c = c(s)$ is a constant that depends only on $s$.
\end{corollary}

\begin{proof}
By induction on $s$, and on $d$ for each $s$. (Recall that
$\N^s_k(m)$ is nonincreasing in $k$ for fixed $s$ and $m$.) The base
case $s=2$ is given by Lemma~\ref{lemma_ns_s1}. For $s = 3$ we have
$R_3(d) = 2d+1$, and the claim is equivalent to
Corollary~\ref{cor_N3km_upper}. Therefore, let $s\ge 4$ be fixed,
and assume the claim holds for $s'<s$.

Let $m_0 = m_0(s)$ be a constant large enough so that\footnote{The
dependence of $m_0$ on $s$ here could be greatly improved with a
slightly more careful analysis.}
\begin{equation*}
m \ge 1+12^s \lceil \log_2 m \rceil^{s^2} \qquad\text{ for all $m\ge
m_0$}.
\end{equation*}
We again work with a slight variant of the inverse Ackermann
function. For this proof define $\widehat \alpha_d(x)$, $d\ge 2$, by
$\widehat\alpha_2(x) = \alpha_2(x) = \lceil \log_2 x \rceil$, and
for $d\ge 3$ by the recurrence
\begin{equation*}
\widehat \alpha_d(x) =
\begin{cases}
1, & \text{if $x\le m_0$};\\
1 + \widehat \alpha_d{\bigl( 12 \widehat \alpha_{d-1}(x)^{s-2}
\bigr)}, & \text{otherwise}.
\end{cases}
\end{equation*}
The functions $\widehat\alpha_d(x)$ are well defined by our choice
of $m_0$. And as before, there exists a constant $c_0$ (depending
only on $s$) such that $| \widehat\alpha_d(x) - \alpha_d(x)|\le c_0$
for all $d$ and $x$.

We will show, by induction on $d$, that there exists a constant
$c_1$ (depending only on $s$) such that, for all $d\ge 2$ and all
$m$, we have
\begin{equation}\label{eq_N_ind_claim}
\N^s_k(m) \le c_1 m \widehat \alpha_d(m)^{s-2} \qquad \text{for
$k\ge R_s(d)$}.
\end{equation}
This is easily seen to imply the claim.

The base case $d = 2$ follows from Corollary~\ref{cor_log}, provided
$c_1$ is chosen large enough. Further, by induction on $s$ we know
there exist constants $c_2$, $c_3$ (depending on $s$), such that
\begin{align*}
\N^{s-1}_k(m) \le c_2 m \widehat \alpha_d(m)^{s-3} \qquad \text{for
$k \ge R_{s-1}(d)$},\\
\N^{s-2}_k(m) \le c_3 m \widehat \alpha_d(m)^{s-4} \qquad \text{for
$k \ge R_{s-2}(d)$},
\end{align*}
for all $d\ge 3$ and all $m$. We choose $c_1$ large enough so that
it also satisfies $c_1 \ge 6c_2$, $c_1 \ge 6c_3$, and
\begin{equation}\label{eq_c1_condition_m0}
c_1 \ge \N^s_{R_s(3)}(m) / m, \qquad \text{for all $m\le m_0$}.
\end{equation}

Now, let $d\ge 3$, and suppose (\ref{eq_N_ind_claim}) holds for
$d-1$. To establish (\ref{eq_N_ind_claim}) for $d$, assume first
that $m\le m_0$. Then, by (\ref{eq_c1_condition_m0}), for all $k\ge
R_s(d)$ we have
\begin{equation*}
\N^s_k(m) \le \N^s_{R_s(3)}(m) \le c_1 m = c_1 m \widehat
\alpha_d(m)^{s-2}.
\end{equation*}
Thus, let $m > m_0$. Apply Recurrence~\ref{rec_nskm} with the
following parameters:
\begin{gather*}
k_1 = R_{s-1}(d), \quad k_2 = R_{s-2}(d), \quad k_3 = R_s(d-1),\\
k = R_s(d), \quad t = 12\widehat\alpha_{d-1}(m)^{s-2}.
\end{gather*}
The last three terms in Recurrence~\ref{rec_nskm} can be bounded as
follows.
\begin{align*}
2\N^{s-1}_{k_1}(t) &\le 2c_2 t \widehat\alpha_d(t)^{s-3} \le
{c_1\over 3} t \widehat \alpha_d(m)^{s-3},\\
\N^{s-2}_{k_2}(t) &\le c_3 t \widehat\alpha_d(t)^{s-4}
\le {c_1\over 6} t \widehat \alpha_d(m)^{s-3},\\
\N^s_{k_3}{\left( 1+ {m\over t}\right)} &\le
\N^s_{k_3}{\left({2m\over t}\right)} \le {2 c_1 m \over t} \widehat
\alpha_{d-1}(m)^{s-2} = {c_1\over 6} m \le {c_1\over 6} m
\widehat\alpha_d(m)^{s-3}.
\end{align*}
Substituting into Recurrence~\ref{rec_nskm} we get
\begin{equation*}
\N^s_k(m) \le {m\over t} \N^s_k(t) + {2 c_1\over 3} m \widehat
\alpha^d(m)^{s-3} + \N^s_k(t) + {c_1\over 2} t \widehat
\alpha_d(m)^{s-3}.
\end{equation*}
But by Lemma~\ref{lemma_ns_s1} we have $\N^s_k(t) \le t^{s-1} \le
\bigl( 12 \lceil \log_2 m\rceil^{s-2} \bigr)^{s-1}$, which is at
most $m$ for $m > m_0$ by our choice of $m_0$. In turn, $m$ is at
most $c_1m/6$, since $c_1\ge 6$.

Similarly, we have ${c_1\over 2} t \widehat \alpha_d(m)^{s-3} \le
c_1 m /6$ for $m>m_0$ by our choice of $m_0$. Thus,
\begin{equation*}
\N^s_k(m) \le {m\over t} \N^s_k(t) + c_1 m \widehat \alpha(m)^{s-3}
\end{equation*}
Letting $\N^s_k(m) = mg(m)$, we get
\begin{equation*}
g(m) \le g(t) + c_1 \widehat\alpha_d(m)^{s-3}
\end{equation*}
Since $\widehat\alpha_d(t) = \widehat\alpha_d(m)-1$, it follows by
induction on $m$ that
\begin{equation*}
g(m) \le c_1 \widehat\alpha_d(m)^{s-2}\qquad \text{for all $m$}.
\end{equation*}
(The base case $m\le m_0$ follows from (\ref{eq_c1_condition_m0}),
and for the induction on $m$ we apply $\bigl( \widehat\alpha_d(m) -
1 \bigr)^{s-2} \le \bigl(\widehat \alpha_d(m) - 1 \bigr) \widehat
\alpha_d(m)^{s-3}$.) Therefore,
\begin{equation*}
\N^s_k(m) \le c_1 m\widehat\alpha_d(m)^{s-2}\qquad \text{for all
$m$}. \qedhere
\end{equation*}
\end{proof}

Let us now study the asymptotic growth of $R_s(d)$ for fixed $s$. We
have
\begin{equation*}
R_3(d) = 2d + 1, \qquad R_4(d) = 5\cdot 2^d - 4d-3.
\end{equation*}
In general, letting $t = \lfloor (s-2) / 2 \rfloor$, we have
\begin{equation}\label{eq_asymp_R}
R_s(d) =
\begin{cases}
2^{(1/t!) d^t \pm O(d^{t-1})}, & \text{$s$ even};\\
2^{(1/t!) d^t \log_2 t \pm O(d^t)}, & \text{$s$ odd}
\end{cases}
\end{equation}
(see Appendix~\ref{app_growth_constants} again).

Lemma \ref{lemma_gral_psi_alphak_bound} now follows from Lemma
\ref{lemma_psi_to_ADS} by applying Corollary \ref{cor_nskm} with $k
= R_s(d)$.

\section{Bounding formation-free sequences}\label{sec_FF}

We now deal with the generalizations of Davenport--Schinzel
sequences described in the Introduction. Recall that the first step
in bounding $\Ex_u(n)$ is Lemma~\ref{lemma_Ex_to_F}, which claims
that $\Ex_u(n) \le \F_{r,s-r+1}(n)$, where $r = \|u\|$ and $s =
|u|$.

\begin{proof}[Proof of Lemma~\ref{lemma_Ex_to_F}]
Suppose $u = u_1 u_2\ldots u_s$, where $1\le u_i \le r$ for each
$i$. We can assume that the symbols in $u$ make their first
appearances in the order $1, 2, \ldots, r$.

Let $s' = s-r+1$, and let $\ell = \ell_1 \ell_2 \cdots \ell_{s'}$ be
an arbitrary $(r,s')$-formation, where each $\ell_j$ is a
permutation of $\{1, \ldots, r\}$. We want to show that $u \subset
\ell$.

Define a partition $u = B_1 B_2 \ldots B_{s'}$ of $u$ into $s'$
blocks as follows: First let each symbol of $u$ constitute its own
block of length $1$. Then, for each $2\le j\le r$, merge the block
that contains the first occurrence of $j$ in $u$ with the block
containing the immediately preceding symbol. The number of blocks
goes down from $s$ to $s'$.

Here is an example of a sequence thus partitioned:
\begin{equation}\label{eq_u_blocks_example}
u = [1][1][12][134][2][4][1][25][5].
\end{equation}
Clearly, each block $B_j$ is an increasing sequence.

Now we are going to define a permutation $\sigma$ on $\{1, \ldots,
r\}$ such that, for each block $B_j$ with $1\le j\le s'$, its image
$\sigma(B_j)$ is a subsequence of $\ell_j$. We do this by examining
the blocks from right to left, and by defining $\sigma$ in the order
$\sigma(r), \sigma(r-1), \ldots, \sigma(1)$. Note that blocks of
length $1$ can be safely ignored.

Suppose we have already dealt with blocks $B_{s'}, B_{s'-1}, \ldots,
B_{j+1}$, and that now is the turn of block $B_j$, where $|B_j| >
1$. Let $k$ be the last symbol in $B_j$. The symbols preceding $k$
in $B_j$ are $k-1,k-2, \ldots$, up to the second symbol of $B_j$.
All these symbols make their first appearance in $u$ in $B_j$. Call
these the ``new" symbols of $B_j$.

Suppose we have already assigned values to $\sigma(k+1), \ldots,
\sigma(r)$ in such a way that, no matter how we assign
$\sigma(1),\ldots, \sigma(k)$, the images $\sigma(B_{j+1}), \ldots,
\sigma(B_{s'})$ will always be subsequences of $\ell_{j+1}, \ldots,
\ell_{s'}$, respectively.

Now consider the symbols of $\ell_j$. Call a symbol of $\ell_j$
``free" if it has not yet been assigned as image $\sigma(i)$ to any
symbol $i$, for $k+1 \le i \le r$.

We scan $\ell_j$ form right to left, considering only its free
symbols, and we assign in a greedy fashion these free symbols as
images $\sigma(k), \sigma(k-1), \ldots$ to $k, k-1, \ldots$ (the
``new" symbols of $B_j$).

After we are done with these assignments, the only symbol of $B_j$
which has not been assigned an image is the first symbol of
$B_j$---call it $b_j$. But no matter how we define $\sigma(b_j)$
later on, we will always have that $\sigma(B_j)$ is a subsequence of
$\ell_j$ (because of our greedy approach).

At the end, the assignment $\sigma(1)$ of $1$ will be forced.

For example, with $u$ is as in (\ref{eq_u_blocks_example}), suppose
that
\begin{equation*}
\ell = \ell_1\ \ell_2\ 32514 \ 35421\ \ell_5\ \ell_6\ \ell_7\ 35142\
\ell_9
\end{equation*}
(where $\ell_1$, $\ell_2$, $\ell_5$, $\ell_6$, $\ell_7$, $\ell_9$ do
not matter). Then, our algorithm will assign $\sigma(5) = 2$,
$\sigma(4) = 1$, $\sigma(3) = 4$, $\sigma(2) = 5$, and finally
$\sigma(1) = 3$. Then the sequence
\begin{equation*}
\sigma(u) = [3][3][35][341][5][1][3][52][2]
\end{equation*}
is a subsequence of $\ell$, as desired.
\end{proof}

\begin{remark}
Lemma~\ref{lemma_Ex_to_F} is not the last word in finding sequences
in formations. For example, consider the sequence $u = abcabca$.
Lemma~\ref{lemma_Ex_to_F} states that $u$ is contained in every
$(3,5)$-formation, but in fact $u$ is contained in every
$(3,4)$-formation: Let $\ell = \ell_1 \ell_2 \ell_3 \ell_4$ be a
$(3,4)$-formation. Suppose $\ell_1 = abc$. Then, if $u$ itself is
not a subsequence of $\ell$, then $\ell_2$ must have $b$ before $a$,
$\ell_3$ must have $c$ before $b$, and $\ell_4$ must have $a$ before
$c$. But then $\ell$ contains the subsequence $cbacbac$.
\end{remark}

\subsection{Bounding the length of formation-free sequences}

Thus, the problem of bounding $\Ex_u(n)$ reduces to that of bounding
$\F_{r,s}(n)$. For completeness, we start by reproducing some simple
bounds from~\cite{klazar92}. We first prove that $\F_{r,s}(n)$ is
finite.

\begin{lemma}[Klazar~\cite{klazar92}]
We have $\F_{r,s}(n) \le s n^r$ for $n\ge r$.
\end{lemma}

\begin{proof}
Let $S$ be an $(r,s)$-formation-free sequence on $n$ distinct
symbols. Partition $S$ from left to right into blocks of length $r$.
Note that each block contains $r$ distinct symbols. Suppose we had
$1 + (s-1){n \choose r}$ complete blocks. Then, by the pigeonhole
principle, there would exist $s$ blocks that have the same set of
$r$ symbols. Such a set of $s$ blocks would be an $(r,s)$-formation.
Contradiction.

Therefore, we must have
\begin{equation*}
|S| < r\left( 1+ (s-1) {n\choose r} \right) \le r s {n \choose r}
\le s n^r. \qedhere
\end{equation*}
\end{proof}

It is also easy to get linear bounds for $\F_{r,2}(n)$ and
$\F_{r,3}(n)$:

\begin{lemma}[Klazar~\cite{klazar92}]\label{lemma_Fr2_Fr3}
We have $\F_{r,2}(n) \le rn$ and $\F_{r,3}(n) \le 2rn$.
\end{lemma}

\begin{proof}
Let $S$ be an $r$-sparse sequence on $n$ distinct symbols. Again
partition $S$ from left to right into blocks of length $r$ (the last
block might be shorter).

If $S$ contains no $(r,2)$-formation then every block must contain
the first occurrence of a symbol, and if $S$ contains no
$(r,3)$-formation, then every block must contain the first \emph{or}
last occurrence of a symbol. Thus, there are at most $n$ blocks in
the first case, and at most $2n$ blocks in the second case.
\end{proof}

\begin{lemma}[Klazar~\cite{klazar92}]\label{lemma_r_sparsify}
Let $S = S_1 S_2 \cdots S_m$ be a sequence which is a concatenation
of $m$ blocks, where each block $S_i$ contains only distinct
symbols. Then $S$ can be made $r$-sparse by deleting at most
$(r-1)(m-1)$ symbols.
\end{lemma}

\begin{proof}
Build an $r$-sparse subsequence $S'$ of $S$ in a greedy fashion, by
scanning $S$ from left to right and adding a symbol from $S$ to $S'$
only if it does not equal any of the last $r-1$ symbols currently in
$S'$. In this way, we will skip at most $r-1$ symbols of each block
$S_i$, except for the first block $S_1$, which we will take
entirely.
\end{proof}

Next, we make a definition analogous to Definition~\ref{def_psi}:

\begin{definition}\label{def_psif}
Given integers $r$, $s$, $m$, and $n$, we denote by
$\psif_{r,s}(m,n)$ the length of the longest $r$-sparse,
$(r,s)$-formation-free sequence on $n$ distinct symbols that can be
partitioned into $m$ or fewer blocks, each block containing only
distinct symbols.
\end{definition}

\begin{remark}
The reader need not be intimidated (more than necessary) by the
double subscript $r,s$ in $\psif_{r,s}(m,n)$. We are never going to
use induction on $r$, only on $s$. Thus, $r$ can be assumed to be
fixed throughout our analysis.
\end{remark}

The following lemma (analogous to Lemma~\ref{lemma_lambda_to_psi})
relates $\F_{r,s}(n)$ to $\psif_{r,s}(m,n)$.

\begin{lemma}\label{lemma_F_to_psif}
Given fixed integers $r$ and $s$, let $\varphi_{r,s-2}(n)$ be a
nondecreasing function of $n$ such that $\F_{r,s-2}(n) \le n
\varphi_{r,s-2}(n)$ for all $n$. Then,
\begin{equation*}
\F_{r,s}(n) \le 2n + \varphi_{r,s-2}(n) \bigl( 2(r-1)n +
\psif_{r,s}(2n,n) \bigr).
\end{equation*}
\end{lemma}

(This constitutes a minor improvement over Klazar~\cite{klazar92},
since Klazar related $\F_{r,s}(n)$ to $\varphi_{r,s-1}(n)$.)

\begin{proof}
Let $S$ be a maximum-length $(r,s)$-formation-free sequence on $n$
symbols. Thus, $|S| = \F_{r,s}(n)$. Partition $S$ from left to right
into subsequences as follows:

Let $S_1$ be the longest prefix of $S$ that is
$(r,s-2)$-formation-free. Let $x_1$ be the symbol following $S_1$ in
$S$. Thus $S_1 x_1$ contains an $(r,s-2)$-formation. Let $S_2$ be
the longest segment of $S$ after $x_1$ which is
$(r,s-2)$-formation-free, let $x_2$ be the symbol following $S_2$ in
$S$, and so on.

We obtain a partition $S = S_1 x_1 S_2 x_2 \ldots x_{m-1} S_m
[x_m]$, where each $S_i$ is a subsequence and each $x_i$ is a symbol
($x_m$ might or might not be present).

Each subsequence $S_i x_i$ must contain either the first or the last
occurrence of some symbol, for otherwise $S$ would contain an
$(r,s)$-formation. Thus, $m\le 2n$.

Let $n_i = \|S_i\|$. Then
\begin{align*}
\F_{r,s}(n) = |S| \le m + \sum_{i=1}^m |S_i| &\le m + \sum_{i=1}^m
\F_{r,s-2}(n_i) \\
&\le 2n + \sum_{i=1}^m n_i \varphi_{r,s-2}(n_i) \le 2n +
\varphi_{r,s-2}(n) \sum_{i=1}^m n_i.
\end{align*}

So we just have to bound $\sum n_i$. Construct a subsequence $S'$ of
$S$ by taking, for each subsequence $S_i$ in the above partition of
$S$, just the first occurrence of each symbol in $S_i$. Thus, $|S'|
= \sum n_i$. Next, using Lemma~\ref{lemma_r_sparsify},
``$r$-sparsify" $S'$ and obtain a sequence $S''$ with $|S''| \ge
|S'| - (r-1)(m-1)$.

Since $S''$ is a subsequence of $S$, it contains no $(r,
s)$-formation. Further, $S''$ is $r$-sparse and partitionable into
$m$ blocks of distinct symbols. Therefore, $|S''| \le
\psif_{r,s}(m,n)$, and so
\begin{equation*}
\sum_{i=1}^m n_i = |S'| \le (r-1)m + |S''| \le 2(r-1)n +
\psif_{r,s}(2n,n).
\end{equation*}
The claim follows.
\end{proof}

We now apply our ``almost-DS" technique to formation-free sequences.
For this, we introduce and analyze ``almost-formation-free"
sequences. The analysis closely parallels the analysis of almost-DS
sequences.

\subsection{Almost-formation-free sequences}

If $S$ is a sequence, we say that $S$ is an \emph{$\AFF_{r,s,k}(m)$}
sequence if $S$ contains no $(r,s)$-formation, can be partitioned
into $m$ of fewer blocks, each composed of distinct symbols, and
each symbol appears at least $k$ times (in $k$ different blocks).

Note that we do not require $r$-sparsity; this is the reason for
calling $S$ ``almost" formation-free.

Let $\Nf_{r,s,k}(m)$ be the maximum possible number of distinct
symbols in an $\AFF_{r,s,k}(m)$ sequence.

We first show the connection between AFF sequences and
$\psif_{r,s}(m,n)$, and then we derive upper bounds for
$\Nf_{r,s,k}(m)$.

\begin{lemma}\label{lemma_psif_to_FFF}
For all $s\ge 2$ and $k$ we have $\psif_{r,s}(m,n) \le k
(\Nf_{r,s,k}(m) + n)$.
\end{lemma}

\begin{proof}
Let $S$ be a maximum-length $r$-sparse, $(r,s)$-formation-free
sequence on $n$ distinct symbols, partitionable into $m$ blocks.
Thus, $|S| = \psif_{r,s}(m,n)$. Let $k \ge 1$ be the specified
parameter.

Transform $S$ into another sequence $S'$ in which each symbol
appears exactly $k$ times as follows. For each symbol $a$, group the
occurrences of $a$ from left to right into ``clusters" of size $k$,
discarding the $<k$ occurrences left at the end. Replace each $a$ in
the $i$-th cluster by a new symbol $a_i$.

If $s\ge 2$, then this does not introduce any $(r,s)$-formations.
(Proof: Call two symbols $a$ and $b$ \emph{disjoint} if every
occurrence of $a$ lies before every occurrence of $b$ or vice-versa.
Note that if $a$ and $b$ are disjoint, they cannot belong to the
same $(r,s)$-formation for $s\ge 2$. Thus, if $S'$ contains an
$(r,s)$-formation, that formation was already present in $S$.)

We deleted at most $kn$ symbols from $S$, and the result $S'$ is an
$\AFF_{r,s,k}(m)$ sequence ($S'$ is not necessarily $r$-sparse, but
this is fine for an AFF sequence). Therefore, $S'$ contains at most
$\Nf_{r,s,k}(m)$ symbols, each one appearing exactly $k$ times, so
it has length at most $k \cdot \Nf_{r,s,k}(m)$. The claim follows.
\end{proof}

\begin{lemma}\label{lemma_AFF_s_2_k_2}
For every $r\ge 2$ we have $\Nf_{r,2,2}(m) = (r-1)(m-1)$.
\end{lemma}

\begin{proof}
For the upper bound, consider $m-1$ ``separators" between the $m$
blocks. We say that a symbol $a$ ``contributes" to all the
separators between the first two occurrences of $a$. Thus each
symbol contributes to at least one separator. If there were $1 +
(r-1)(m-1)$ symbols, then there would exist a separator with at
least $r$ contributions, which would lead to the existence of an
$(r,2)$-formation.

For the lower bound, create $m$ blocks, and create $n = (r-1)(m-1)$
different symbols partitioned into $m-1$ sets $A_1, \ldots, A_{m-1}$
of $r-1$ symbols each. Make two copies of each $A_i$, and put one
copy at the end of block $i$ and one copy at the beginning of block
$i+1$. We get a sequence with the desired properties.
\end{proof}

\begin{lemma}\label{lemma_FFF_s_times}
For every fixed $r\ge 2$ and $s\ge 3$ we have $\Nf_{r,s,s}(m) \le
(r-1){ m - 2 \choose s-2} = O(m^{s-2})$.
\end{lemma}

\begin{proof}
Suppose for a contradiction that there is an $\AFF_{r,s,s}(m)$
sequence with $1 + (r-1){ m-2 \choose s-2}$ distinct symbols.
Consider the $s-2$ middle occurrences of each symbol. They fall on
$s-2$ out of $m-2$ different blocks. Therefore, by the pigeonhole
principle, there exist $r$ symbols whose $s-2$ middle occurrences
all fall in the same $s-2$ blocks. This leads to the existence of an
$(r,s)$-formation in the given sequence. Contradiction.
\end{proof}

\begin{recurrence}\label{rec_AFF_log}
We have
\begin{equation*}
\Nf_{r,s,2k-1}(2m) \le 2\Nf_{r,s,2k-1}(m) + 2\Nf_{r,s-1,k}(m).
\end{equation*}
\end{recurrence}

The proof is exactly parallel to that of Recurrence~\ref{rec_log}.

\begin{corollary}
For fixed $r\ge 2$ and $s\ge 3$, if we let $k = 2^{s-2} + 1$, then
\begin{equation*}
\Nf_{r,s,k}(m) = O{\bigl( m (\log m)^{s-3} \bigr)}
\end{equation*}
(where the constant implicit in the $O$ notation might depend on $r$
and $s$).
\end{corollary}

\begin{recurrence}\label{rec_AFF_gral}
Let $r\ge 2$ and $s\ge 3$ be fixed. Let $k_1$, $k_2$, $k_3$, and $k$
be integers satisfying $k = k_2 k_3 + 2k_1 - 3k_2 - k_3 + 2$. Then,
\begin{equation*}
\Nf_{r,s,k}(m) \le \left(1+{m\over t}\right) \Bigl( \Nf_{r,s,k}(t) +
2 \Nf_{r,s-1,k_1}(t) + \Nf_{r,s-2,k_2}(t) \Bigr) +
\Nf_{r,s,k_3}{\left(1+ { m \over t } \right)},
\end{equation*}
where $t$ is a free parameter.
\end{recurrence}

The proof exactly parallels that of Recurrence~\ref{rec_nskm}. The
corollary is almost the same as Corollary~\ref{cor_nskm}; there is
just a shift of $1$ in the index $s$:

\begin{corollary}\label{cor_Nrsk_upper}
Let $R_s(d)$ be the sequences defined in Corollary~\ref{cor_nskm}.
Then, for every $s\ge 3$ and $d\ge 2$, if $k \ge R_{s-1}(d)$ then
\begin{equation*}
\Nf_{r,s,k}(m) \le c m \alpha_{d}(m)^{s-3} \qquad \text{for all
$m\ge k$}.
\end{equation*}
Here, $c = c(r,s)$ is a constant that depends only on $r$ and $s$.
\end{corollary}

Combining Corollary~\ref{cor_Nrsk_upper} with
Lemma~\ref{lemma_psif_to_FFF}, we obtain:

\begin{corollary}\label{cor_psif_alphak_bound}
Let $s\ge 4$. Then, for all $r$, $m$, and $n$ we have
\begin{equation*}
\psif_{r,s}(m,n) \le C_{r,s,d}\bigl( m\alpha_d(m)^{s-3} + n \bigr)
\qquad \text{for all $d$},
\end{equation*}
for some constants $C_{r,s,d}$ of the form
\begin{equation*}
C_{r,s,d} =
\begin{cases}
2^{(1/t!) d ^t \pm O{\left( d^{t-1} \right)}}, & \text{$s$ odd};\\
2^{(1/t!) d ^t \log_2 d \pm O{\left( d^t \right)}}, & \text{$s$
even};
\end{cases}
\end{equation*}
where $t = \lfloor (s-3) / 2 \rfloor$.
\end{corollary}

We can finally prove our upper bounds for $\F_{r,s}(n)$.

\begin{proof}[Proof of Theorem~\ref{thm_formation_free_upper}]
Take $d = \alpha(m)$ in Corollary~\ref{cor_psif_alphak_bound}, then
substitute into Lemma~\ref{lemma_F_to_psif}, bounding $\varphi_{r,
s-2}(n)$ by induction on $s$. Use the base cases $\F_{r,2}(n),
\F_{r,3}(n) = O(n)$ (by Lemma~\ref{lemma_Fr2_Fr3}). (As before,
$\varphi_{r, s-2}(n)$ contributes only to lower-order terms in the
exponent.)
\end{proof}

\section{The lower bound construction for $s = 3$}\label{sec_constr_3}

The rest of this paper deals with lower bounds for
Davenport--Schinzel sequences. In this section we prove
Theorem~\ref{thm_lambda3_lower} by constructing, for every $n$, a
Davenport--Schinzel sequence of order $3$ on $n$ distinct symbols
with length at least $2n \alpha(n) - O(n)$.

For this purpose, we first define a two-dimensional array of
sequences $Z_d(m)$, for $d,m \ge 1$, with the following properties:
\begin{itemize}
\item
Each symbol in $Z_d(m)$ appears exactly $2d+1$ times.

\item
$Z_d(m)$ contains no forbidden alternation $ababa$. (We do not
preclude the presence of adjacent repeated symbols in $Z_d(m)$.)

\item
$Z_d(m)$ is partitioned into \emph{blocks}, where each block
contains only distinct symbols. Some of the blocks in $Z_d(m)$ are
\emph{special blocks}. Each symbol in $Z_d(m)$ makes its first and
last occurrences in special blocks. Furthermore, the special blocks
are entirely composed of first and last occurrences of symbols
(there might be \emph{both} first and last occurrences in the same
special block). Moreover, each special block in $Z_d(m)$ has length
exactly $m$.

\item
For $d\ge 2$, each special block is surrounded by regular blocks on
both sides, and \emph{no} regular block is surrounded by special
blocks on both sides. For the former property, we place empty
regular blocks at the beginning and end of $Z_d(m)$, for $d\ge 2$.
\end{itemize}

In what follows, we enclose regular blocks by $(\,)$'s, and special
blocks by $[\,]$'s.

The base cases of the construction are as follows: For $d=1$, we let
\begin{equation*}
Z_1(m) = [1 2 \ldots m] (m \ldots 2 1) [1 2 \ldots m].
\end{equation*}
$Z_1(m)$ contains three blocks of length $m$; the first and last
ones are special blocks. Note that each symbol appears exactly three
times, as required. Also note that $Z_1(m)$ contains no alternation
$ababa$.

For $m=1$ and $d\ge 2$ we let
\begin{equation*}
Z_d(1) = (\,) [1] (1) (1) \ldots (1) [1] (\,),
\end{equation*}
with $2d+1$ ones. Each symbol constitutes its own block; the first
and the last nonempty blocks are special. Note that these special
blocks have length $1$, as required. At the beginning and end there
are regular blocks of length zero.

Denote by $S_d(m)$ the number of special blocks in $Z_d(m)$.

\paragraph{The recursive construction.}

For $d,m \ge 2$, we construct $Z_d(m)$ recursively as follows. Let
$Z' = Z_d(m-1)$. Let $f = S_d(m-1)$ be the number of special blocks
in $Z'$, and let $Z^* = Z_{d-1}(f)$. Thus, the special blocks in
$Z^*$ have length $f$. Let $g = S_{d-1}(f)$ be the number of special
blocks in $Z^*$.

Create $g$ copies of $Z'$, each copy using ``fresh" symbols which do
not occur in $Z^*$ nor in any preceding copy of $Z'$. Thus, we have
one copy of $Z'$ for each special block in $Z^*$. Furthermore, each special
block in $Z^*$ has as many symbols as there are special blocks in
the corresponding copy of $Z'$.

Let $C_i$ be the $i$-th special block in $Z^*$, and let $Z'_i$ be
the $i$-th copy of $Z'$. Let $a$ be the $\ell$-th symbol in $C_i$,
and let $D_\ell$ be the $\ell$-th special block in $Z'_i$. We
duplicate $a$ into $aa$, and we insert the $aa$ into $Z'_i$ as
follows:

If the $a$ in $C_i$ is the first $a$ in $Z^*$, then the first of the
two $a$'s falls at the end of $D_\ell$ and the second $a$ falls at
the beginning of the block after $D_\ell$. Otherwise, if the $a$ in $C_i$
is the last $a$ in $Z^*$, then the first of the two $a$'s falls at
the end of the block before $D_\ell$ and the second $a$ falls at the
beginning of $D_\ell$. (Recall that $D_\ell$ is surrounded by
regular blocks in $Z'_i$.)

Since no regular block in $Z'_i$ is surrounded by special blocks on
both sides, it follows that no block in $Z'_i$ receives more than
one symbol from $Z^*$. Thus, even after the insertions, no block in
$Z'_i$ has repeated symbols.

\begin{figure}
\centerline{\includegraphics{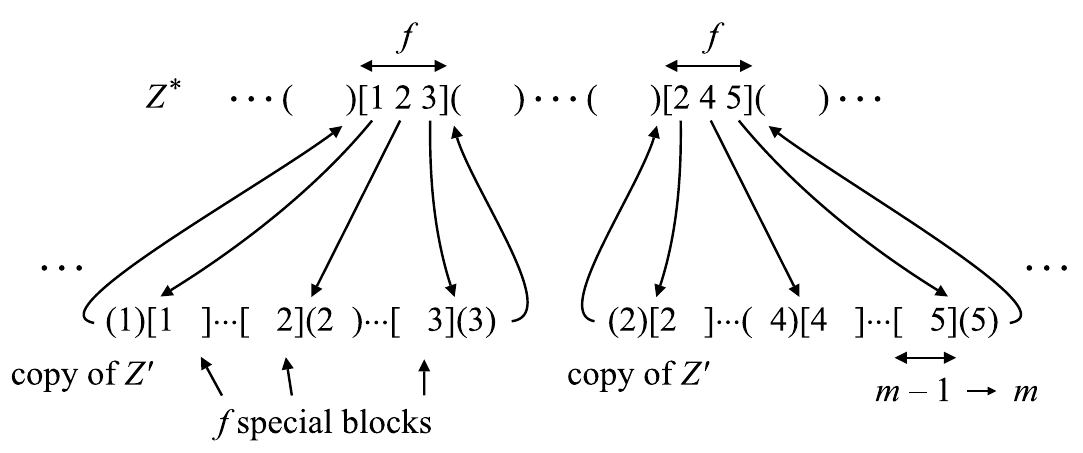}}
\caption{\label{fig_construction_3}Construction of $Z_d(m)$ from
$Z^*$ and many copies of $Z'$. Two special blocks of $Z^*$ are
depicted. In the left one, the symbol $1$ makes its last occurrence,
and symbols $2$, $3$ make their first occurrence. In the right
block, symbols $2$ and $4$ make their last occurrence, and symbol
$5$ makes its first occurrence.}
\end{figure}

After these insertions, at the place in $Z^*$ where the block $C_i$
used to be there is now a hole. We insert $Z'_i$ (with its extra
symbols) into this hole. After doing this for all special blocks
$C_i$ in $Z^*$, we obtain the desired sequence $Z_d(m)$. See
Figure~\ref{fig_construction_3}.

It is easy to check that every symbol in $Z_d(m)$ has multiplicity
$2d+1$: The symbols of the copies of $Z'$ already had multiplicity
$2d+1$, and the symbols of $Z^*$ had their multiplicity increased
from $2d-1$ to $2d+1$.

It is also clear that each symbol makes its first and last
occurrences in special blocks, that the special blocks in $Z_d(m)$
contain only first and last occurrences, and that their length
increased from $m-1$ to $m$. Furthermore, every special block is
surrounded by regular blocks on both sides, and no regular block is
surrounded by special blocks on both sides. And $Z_d(m)$ contains
empty regular blocks at the beginning and at the end.

\paragraph{No $ababa$.}

Let us now verify that $Z_d(m)$ contains no alternation $ababa$ of
length $5$. Assume by induction that this is true for the component
sequences $Z'$ and $Z^*$.

Suppose for a contradiction that $Z_d(m)$ contains an alternation
$ababa$. The symbols $a$ and $b$ cannot come from the same copy of
$Z'$, by induction, and they cannot come from different copies of
$Z'$, since they would not alternate at all.

Further, $a$ and $b$ cannot both come from $Z^*$: By the induction
assumption, $Z^*$ contains no forbidden alternation. And the
duplications of symbols $a \to aa$ cannot create a forbidden
alternation, since the two $a$'s end up being adjacent in $Z_d(m)$.

Next, suppose that $a$ comes from a copy of $Z'$ and $b$ comes from
$Z^*$. Then this copy of $Z'$ received two non-adjacent $b$'s. But
this is impossible by construction: Our copy of $Z'$ received
symbols from a single special block of $Z^*$, which contained at
most one $b$. This $b$ was duplicated into two \emph{adjacent}
copies $bb$.

Finally, suppose that $a$ comes from $Z^*$ and $b$ comes from a copy
of $Z'$. Then this copy of $Z'$ received an $a$ that is neither the
first nor the last $a$ in $Z^*$. This is also a contradiction.

\begin{remark}
The above construction shares some similarities with an earlier
construction by Komj\'ath~\cite{komjath}.
\end{remark}

\subsection{Analysis}\label{subsec_3_analysis}

Recall that $S_d(m)$ denotes the number of special blocks in
$Z_d(m)$. We define a few other quantities related to $Z_d(m)$:
\begin{itemize}
\item
$N_d(m) = \|Z_d(m)\|$ denotes the number of distinct symbols in
$Z_d(m)$.

\item
$L_d(m) = |L_d(m)|$ denotes the length of $Z_d(m)$.

\item
$M_d(m)$ denotes the total number of blocks (regular and special) in
$Z_d(m)$.

\item
We let $X_d(m) = M_d(m) / S_d(m)$. Thus, $X_d(m)^{-1}$ is the
fraction of blocks in $Z_d(m)$ that are special.

\item
We let $V_d(m) = L_d(m) / M_d(m)$ denote the average block length in
$Z_d(m)$.
\end{itemize}
Note that
\begin{align}
N_d(m) &= {1\over 2} m S_d(m),\label{eq_Ndm}\\
L_d(m) &= (2d+1) N_d(m) = \left(d + {1\over 2} \right) m
S_d(m).\label{eq_Ldm}
\end{align}
Equation (\ref{eq_Ndm}) follows from the fact that each symbol
appears in two special blocks, and each special block contains $m$
symbols. Equation (\ref{eq_Ldm}) follows from the fact that each
symbol appears $2d+1$ times in $Z_d(m)$.

Theorem~\ref{thm_lambda3_lower} follows from the following facts:

\begin{lemma}\label{lemma_lambda3_lower_facts}
The quantity $N_d(m)$ experiences Ackermann-like growth.
Specifically, there exists a small absolute constant $c$ such that
\begin{equation}\label{eq_N_A_sandwich}
A_d(m) \le N_d(m) \le A_d(m+c)
\end{equation}
for all $d\ge 3$ and all $m \ge 2$.

We also have $X_d(m) \le 2d + 1$ and $V_d(m) \ge m/2$ for all $d$
and all $m$.
\end{lemma}

Let us first see how this lemma implies
Theorem~\ref{thm_lambda3_lower}.

\begin{proof}[Proof of Theorem~\ref{thm_lambda3_lower}]
Diagonalize by taking the sequences $Z^*_d = Z_d(d)$ for $d = 1, 2,
3, \ldots$. Let $N^*_d = N_d(d)$, $L^*_d = L_d(d)$, and $V^*_d =
V_d(d)$.

By (\ref{eq_N_A_sandwich}) and (\ref{eq_another_rec_A}) we have
$N^*_d \le A_d(d+c) \le A_d{\bigl( A(d+1) \bigr)} = A(d+2)$. Thus,
\begin{equation}\label{eq_N*_A_sandwich}
A(d) < N^*_d \le A(d+ 2)
\end{equation}
for all $d\ge 4$. Thus, by (\ref{eq_alpha_A}),
\begin{equation}\label{eq_d_alpha_sanwich}
\alpha{\left( N^*_d \right)} - 2 \le d < \alpha{\left( N^*_d
\right)}
\end{equation}
for $d\ge 4$, and so, by (\ref{eq_Ldm}),
\begin{equation*}
L^*_d \ge 2 N^*_d \cdot \alpha{\left( N^*_d \right)} - O{\left(
N^*_d\right)}.
\end{equation*}
The sequences $Z^*_d$ are not necessarily Davenport--Schinzel
sequences, since they might have adjacent repeated symbols.
Therefore, create sequences $Z'_d$ by removing adjacent repetitions
from $Z^*_d$. Since we delete at most one symbol per block, the
length of $Z^*_d$ decreases by at most a $1/ V^*_d$ fraction. But by
Lemma \ref{lemma_lambda3_lower_facts} this ratio tends to zero with
$d$ (this is why we diagonalized). Specifically, the length of
$Z'_d$ is
\begin{equation}\label{eq_L'd_lower}
L'_d \ge L^*_d \left( 1 - {1\over V^*_d} \right) \ge L^*_d \left( 1
- {2\over d} \right) \ge 2N^*_d \cdot \alpha{\left( N^*_d \right)} -
O{\left( N^*_d \right)}.
\end{equation}

We have just proven that $\lambda_3(n) \ge 2n \alpha(n) - O(n)$ for
$n$ of the form $n = N^*_d$. We just have to interpolate to
intermediate values of $n$. Given $n$, let $d = d(n)$ be the unique
integer such that
\begin{equation*}
N^*_d < N^*_{d+1} \le n < N^*_{d+2}.
\end{equation*}
It follows, by applying (\ref{eq_d_alpha_sanwich}) twice, that
\begin{equation}\label{eq_alpha_n_alpha_N*}
\alpha(n) \le \alpha{\left( N^*_{d+2} \right)} \le d+4 <
\alpha{\left( N^*_d \right)} +4.
\end{equation}
Also, by the rapid growth of $N^*_d$ in $d$, we certainly have
\begin{equation}\label{eq_N*_sqrt_n}
N^*_d \le \sqrt{ N^*_{d+1}} \le \sqrt n
\end{equation}
for $d\ge 4$.

We now concatenate many copies of $Z'_d$ with disjoint sets of
symbols, making sure we do not have more than $n$ distinct symbols
altogether. Specifically, we let $t = \lfloor n / N^*_d \rfloor$,
and we let $Z''(n)$ be a concatenation of $t$ copies of $Z'_d$ with
disjoint sets of symbols.

By (\ref{eq_L'd_lower}), (\ref{eq_alpha_n_alpha_N*}), and
(\ref{eq_N*_sqrt_n}), it follows that the length of $Z''(n)$ is
\begin{equation*}
L''(n) = t L'_d \ge \left( {n \over N^*_d} - 1 \right) \bigl( 2
N^*_d \cdot \alpha{\left( N^*_d \right)} - O{\left( N^*_d\right)}
\bigr) = 2n \alpha(n) - O(n).
\end{equation*}
Since $\lambda_3(n) \ge L''(n)$, the bound follows.
\end{proof}

All that remains is to prove Lemma~\ref{lemma_lambda3_lower_facts}.

\begin{proof}[Proof of Lemma~\ref{lemma_lambda3_lower_facts}]

The quantity $S_d(m)$ is given recursively by
\begin{align}
S_1(m) &= 2; \notag\\
S_d(1) &= 2; \notag\\
S_d(m) &= fg = S_d(m-1) S_{d-1}\bigl( S_d(m-1) \bigr), \qquad
\text{for $d,m \ge 2$}. \label{eq_Sdm_rec}
\end{align}
In particular, we have $S_2(m) = 2^m = A_2(m)$, and $S_d(2) = 2^d$.

It is not hard to show (see Appendix~\ref{app_ack_like}) that there
exists a small constant $c_0$ such that
\begin{equation}\label{eq_S_A_sandwich}
A_d(m) \le S_d(m) \le A_d(m+c_0)
\end{equation}
for all $d\ge 2$ and all $m$. Then, by (\ref{eq_Ndm}) we have, for
$d\ge 3$, $m\ge 2$,
\begin{equation*}
S_d(m) \le N_d(m) \le S_d(m)^2 \le S_d(m+1),
\end{equation*}
so (\ref{eq_N_A_sandwich}) follows with $c = c_0 + 1$.

Regarding $M_d(m)$, we have
\begin{align*}
M_1(m) &= 3;\\
M_d(1) &= 2d+3, \quad \text{for $d\ge 2$},
\end{align*}
counting the empty blocks at the ends of $Z_d(1)$. And for $d,m \ge
2$, we have
\begin{align}
M_d(m) &= g M_d(m-1) + M_{d-1}(f) - g \notag\\
&= S_{d-1}\bigl( S_d(m-1) \bigr) \bigl(M_d(m-1)-1 \bigr) +
M_{d-1}\bigl( S_d(m-1) \bigr) \label{eq_Mdm_rec}
\end{align}
(since the $g$ special blocks of $Z^*$ disappear). In particular, we
have $M_2(m) = 2^{m+2} - 1$, and $M_d(2) = 2^{d+1}d - 1$.

Let us now examine $X_d(m) = M_d(m) / S_d(m)$. We have
\begin{align*}
X_1(m) &= 3/2,\\
X_2(m) &= 4 - 2^{-m}, \\
X_d(1) &= d+3/2, \qquad \text{for $d \ge 2$},\\
X_d(2) &= 2d - 2^{-d}, \qquad \text{for $d \ge 2$}.
\end{align*}
In general, dividing (\ref{eq_Mdm_rec}) by (\ref{eq_Sdm_rec}),
\begin{equation}\label{eq_Xdm_rec}
X_d(m) = X_d(m-1) + {X_{d-1}\bigl( S_d(m-1) \bigr) - 1\over
S_d(m-1)}.
\end{equation}

We now prove by induction that $X_d(m) \le 2d+1$ for all $d$ and
$m$. The claim has been verified for $d \le 2$ and for $m\le 2$, so
assume $d,m \ge 3$. By (\ref{eq_Xdm_rec}) and using induction on
$d$, we have
\begin{equation*}
X_d(m) \le X_d(m-1) + {2d-2 \over S_d(m-1)},
\end{equation*}
so
\begin{equation*}
X_d(m) \le X_d(2) + (2d-2) \sum_{m=2}^\infty S_d(m)^{-1} = 2d -
2^{-d} + (2d-2) \sum_{m=2}^\infty S_d(m)^{-1}.
\end{equation*}
It is easily checked that, for $d\ge 3$,
\begin{equation*}
\sum_{m=2}^\infty S_d(m)^{-1} \le 2S_d(2)^{-1} = 2^{1-d} \le {1\over
2d-2}.
\end{equation*}
It follows that $X_d(m) \le 2d+1$, as desired.

Finally, let us consider $V_d(m)$. By (\ref{eq_Ldm}) we have
\begin{equation*}
V_d(m) = {L_d(m) \over M_d(m)} = \left( d + {1\over 2} \right) {m
\over X_d(m)} \ge {m\over 2}. \qedhere
\end{equation*}
\end{proof}

\begin{remark}
The coefficient $2$ in our bound for $\lambda_3(n)$ comes from the
fact that each symbol appears roughly $2d$ times in $Z_d(m)$. In
previous constructions \cite{WS88,komjath,DS_book} each symbol
appears only $d \pm O(1)$ times in the equivalent sequence. Sharir
and Agarwal~\cite{DS_book} lost an additional factor of $2$ in the
interpolation step; we avoided this loss in the proof of
Theorem~\ref{thm_lambda3_lower} by letting $Z''(n)$ consist of many
copes of $Z'_d$, instead of using $Z'_{d+1}$ (which would have been
a more obvious choice).
\end{remark}

\subsection{Lower bound for the number of symbols in almost-DS
sequences of order $3$} \label{subsec_ADS_3}

In Section \ref{sec_psi_new} we introduced the notion of almost-DS
sequences. We derived an upper bound on the maximum number
$\N^s_k(m)$ of distinct symbols of an $\ADS^s_k(m)$-sequence, and we
used this upper bound to bound $\lambda_s(n)$.

But the problem of $\ADS^s_k(m)$-sequences is interesting in its own
right, so one might naturally wonder about matching lower bounds for
$\N^s_k(m)$.

It turns out that the construction $Z_d(m)$ described in this
section also provides a roughly-matching lower bound for
$\N^3_d(x)$. We just have to change our point of view: Instead of
taking a diagonal (namely, $Z_d(d)$), we take the \emph{rows} of the
construction (meaning, $Z_d(m)$ for fixed $d$).

\begin{lemma}\label{lemma_ADS_3_lower}
For every fixed $d\ge 2$ we have $\N^3_{2d+1}(x) = \Omega{\bigl(
{1\over d} x\alpha_d(x) \bigr)}$.
\end{lemma}

\begin{proof}
For every $m\ge 1$, the sequence $Z_d(m)$ is an
$\ADS^3_{2d+1}(x_m)$-sequence for $x_m = M_d(m)$. Let $n_m = N_d(m)$
be the number of distinct symbols in $Z_d(m)$.

By the definition of $X_d(m)$, and by applying
Lemma~\ref{lemma_lambda3_lower_facts} and then
(\ref{eq_S_A_sandwich}), we have
\begin{equation}\label{eq_x_m_A_d}
\begin{split}
x_m = M_d(m) = X_d(m) S_d(m) &\le (2d+1) S_d(m) \\
&\le (2d+1) A_d(m+c_0) \le A_d(m+c_0+1).
\end{split}
\end{equation}
Thus, by (\ref{eq_alphak_Ak}) we have $m \ge \alpha_d(x_m) - c_0 -
1$. Therefore, by (\ref{eq_Ldm}), applying
Lemma~\ref{lemma_lambda3_lower_facts} again, and applying (\ref{eq_x_m_A_d}),
\begin{equation*}
n_m = N_d(m) = {L_d(m) \over 2d+1} = {V_d(m) M_d(m) \over 2d+1} \ge
{m M_d(m) \over 4d+2} = \Omega{\left( {1\over d} x_m \alpha_d(x_m)
\right)}.
\end{equation*}
We interpolate to intermediate values of $x$ (for $x_m \le x <
x_{m+1}$) as we did above, in Section~\ref{subsec_3_analysis}.
\end{proof}

Thus, for odd $d$ the bounds for $\N^3_d(m)$ are quite tight (they
leave a multiplicative gap of $O(d)$). For even $d$ the bounds are
not so tight---they are obtained by applying $\N^3_{d+1}(m) \le
\N^3_d(m) \le \N^3_{d-1}(m)$.

Lemma~\ref{lemma_ADS_3_lower} automatically yields a lower bound for
$\Nf_{r,4,k}(m)$: A sequence that does not contain $ababa$ cannot
contain an $(r,4)$-formation for any $r\ge 2$; further, as Adamec,
Klazar, and Valtr~\cite{AKV92} showed, an $r$-sparse, $u$-free
sequence can be made $r'$-sparse for $r'>r$ at the cost of shrinking
the sequence by at most a constant factor.

\section{The lower-bound construction for $s\ge 4$ even}\label{sec_constr_even}

In this section we present a construction that achieves the lower
bounds (\ref{eq_lambda_lower}). This is a simpler variant of the
construction of Agarwal, Sharir, and Shor \cite{ASS89, DS_book} that
achieves the same bounds.

We first construct a family of sequences $S^s_k(m)$ for $s\ge 2$
even, $k\ge 0$, and $m\ge 1$. For all $s\ge 4$, $m\ge 2$, the
sequences $S^s_k(m)$ are Davenport--Schinzel sequences of order $s$.

The sequences $S^s_k(m)$ are highly regular; they satisfy the
following properties:

\begin{itemize}
\item
$S^s_k(m)$ is a concatenation of blocks of length $m$, where each
block contains $m$ distinct symbols. (For $s=2$ or $m=1$ there are
adjacent repeated symbols at the interface between blocks, but only
in these cases.)

\item
$S^s_k(m)$ does not contain any forbidden alternation $abab\ldots$
of length $s+2$, for any distinct symbols $a\neq b$. Thus, for $s\ge
4$, $m\ge 2$, the sequence $S^s_k(m)$ is a Davenport--Schinzel
sequence of order $s$.

\item
All symbols in $S^s_k(m)$ occur with the same multiplicity
$\mu_s(k)$, which depends only on $s$ and $k$. Further, for $s\ge 4$
each symbol in $S^s_k(m)$ makes all its appearances in the same
position within the blocks, and no two symbols $a, b$ appear
together in more than one block.
\end{itemize}

\subsection{The construction}

For $s=2$, the sequences $S^2_k(m)$ are given (independently of $k$)
by
\begin{equation*}
S^2_k(m) = 1 2 \ldots m\ m \ldots 2 1.
\end{equation*}
$S^2_k(m)$ consists of two blocks of length $m$, and each symbol
occurs with multiplicity $\mu_2(k) = 2$. Clearly, $S^2_k(m)$
contains no forbidden alternation $abab$.

The construction for general $s\ge 4$ is as follows. For $k=0$, we
let $S^s_0(m)$ consist of a single block of length $m$:
\begin{equation}\label{eq_Ss0m}
S^s_0(m) = 1 2 \ldots m.
\end{equation}
Thus, $\mu_s(0) = 1$.

For general $k\ge 1$, we proceed as follows. The sequence $S^s_k(1)$
consists of
\begin{equation}\label{eq_rec_mu}
\mu_s(k) = \mu_{s-2}(k-1) \mu_s(k-1)
\end{equation}
copies of the symbol $1$, each forming by itself a block of length
one. Equation (\ref{eq_rec_mu}), together with the bounding cases
$\mu_2(k) = 2$ and $\mu_s(0) = 1$ for $s\ge 4$, gives the recursive
definition of $\mu_s(k)$.

For $m\ge 2$, the sequence $S^s_k(m)$ is constructed inductively on
the lexicographic order of the triples $(s,k,m)$, using three
previously created sequences as components.

The first sequence is $S' = S^s_k(m-1)$; note that $S'$ contains
blocks of length $m-1$. Let $f$ be the number of blocks in $S'$.

The second sequence is $\overline S = S^{s-2}_{k-1}(f)$. Thus,
$\overline S$ contains blocks of length $f$. Let $g = \left\|
\overline S \right\|$ be the number of distinct symbols in
$\overline S$.

The third and final sequence is $S^* = S^s_{k-1}(g)$. Thus, $S^*$
contains blocks of length $g$.

Transform the sequence $S^*$ into a sequence $\widehat S^*$ by
replacing each block in $S^*$ by a copy of $\overline S$ with the
same set of $g$ symbols, making their first appearances in the same
order as in the replaced block. Note that $\widehat S^*$ contains
blocks of length $f$. Further, the multiplicity of each symbol in $\widehat
S^*$ is the product of the symbol multiplicities in $\overline S$ and $S^*$; by incduction this multiplicity equals
\begin{equation*}
\mu_{s-2}(k-1) \mu_s(k-1) = \mu_s(k).
\end{equation*}
Let $h$ be the number of blocks in $\widehat S^*$.

Now, create $h$ copies of $S'$, each copy using ``fresh" symbols
which do not occur in $\widehat S^*$ nor in any preceding copy of
$S'$, and concatenate them into a sequence $S''$. Note that $S''$
contains $fh$ blocks of length $m-1$, while $\bigl| \widehat S^*
\bigr| = fh$.

Insert each symbol of $\widehat S^*$ in order at the end of each
block of $S''$. Thus, each component sequence $S'$ in $S''$,
containing $f$ blocks, receives the $f$ distinct symbols of a block
in $\widehat S^*$. The resulting sequence is the desired $S^s_k(m)$.
Note that it contains blocks of length $m$, and, by induction and
construction, each symbol in it has multiplicity $\mu_s(k)$. See
Figure~\ref{fig_even_constr}.

\begin{figure}
\centerline{\includegraphics{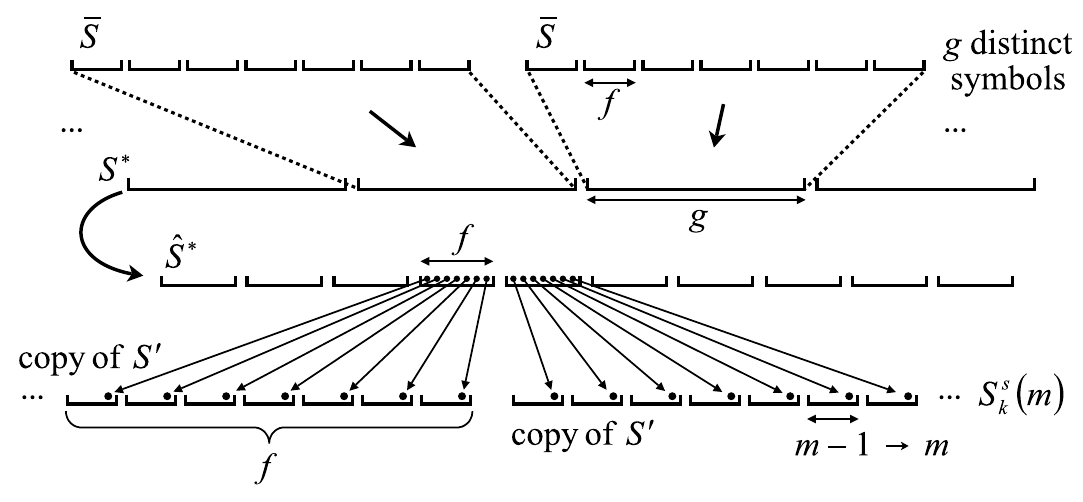}}
\caption{\label{fig_even_constr}The recursive construction of
$S^s_k(m)$. The sequence $\widehat S^*$ is the result of replacing
each block of $S^*$ by a copy of $\overline S$. Each block of
$\widehat S^*$ is then distributed among the $f$ blocks of a single
copy of $S'$.}
\end{figure}

Letting $t = s/2-1$, we have
\begin{equation}\label{eq_bd_mu}
\mu_s(k) = 2^{k\choose t} = 2^{(1/t!)k^t - O(k^{t-1})},
\end{equation}
if we take $s$ to be a constant.

\subsection{Correctness of the construction}

We now prove that, for $s\ge 4$, $m\ge 2$, the sequences $S^s_k(m)$
are indeed Davenport--Schinzel sequences of order $s$.

Let us first recall some important properties of the construction:
\begin{itemize}
\item
The last symbol in each block of $S^s_k(m)$ comes from $\widehat
S^*$ (which has the same set of symbols as $S^*$), while every other
symbol in $S^s_k(m)$ comes from a copy of $S'$.

\item
The copies of $S'$ have pairwise disjoint sets of symbols, which are
also disjoint from the set of symbols of $\widehat S^*$.

\item
When merging $S''$ and $\widehat S^*$ to form $S^s_k(m)$, each copy
of $S'$ in $S''$ receives the $f$ \emph{distinct} symbols of a block
of $\widehat S^*$.
\end{itemize}

The following lemma is easily proven by induction using the above
properties:

\begin{lemma}\label{lemma_easy_properties}
The sequence $S^s_k(m)$ satisfies the following properties:
\begin{enumerate}
\item
For $s\ge 4$, each symbol in the sequence makes all its appearances
in the same position within the blocks.

\item
For $s\ge 4$, $m\ge 2$, there are no adjacent repeated symbols.

\item
For $s\ge 4$, no two symbols of $S^s_k(m)$ appear together in more
than one block.
\end{enumerate}
\end{lemma}

For each symbol $a$ in $S^s_k(m)$, call the \emph{depth} of $a$ the
position within the blocks in which $a$ always appears in
$S^s_k(m)$. This notion is well-defined by the above lemma. Thus,
the symbols that come from copies of $S'$ have depth between $1$ and
$m-1$, while the symbols that come from $\widehat S^*$ have depth
$m$.

The following Lemma is also pretty straightforward:

\begin{lemma}\label{lemma_at_most_ababa}
Symbols at different depths in $S^s_k(m)$ make alternations of
length at most $5$.
\end{lemma}

\begin{proof}
By induction. The claim is clearly true if $s = 2$, $k = 0$, or $m =
1$. Thus, let $s\ge 4$, $k\ge 1$, and $m\ge 2$. Let $a$ and $b$ be
two symbols at different depths in $S^s_k(m)$.

If both $a$ and $b$ have depth at most $m - 1$, then they either
come from the same copy of $S'$, in which case the claim follows by
induction, or else they come from different copies of $S'$, in which
case they do not alternate at all.

Thus, suppose one symbol, say $a$, has depth $m$ (so it comes from
$\widehat S^*$), while the other symbol, $b$, has depth at most
$m-1$ (so it comes from a copy of $S'$).

The copy of $S'$ to which $b$ belongs receives at most one $a$ from
$\widehat S^*$. In the worst case, this $a$ is surrounded by $b$'s
from our copy of $S'$, and this copy of $S'$ is in turn surrounded
by other $a$'s from $\widehat S^*$. Thus the longest alternation we
can get is $ababa$.
\end{proof}

The main issue is to show that $S^s_k(m)$ contains no forbidden
alternating subsequence of length $s+2$. For this, we prove by
induction that $S^s_k(m)$ satisfies a stronger property.

\begin{lemma}
The sequence $S^s_k(m)$ satisfies the following properties:

\begin{enumerate}
\item\label{property_no_forbid}
$S^s_k(m)$ contains no forbidden alternation $abab\ldots$ of length
$s+2$.

\item\label{property_replace_no_forbid}
Furthermore, if each block $B$ in $S^s_k(m)$ is replaced by a
sequence $T(B)$ on the same set of symbols as $B$, such that $T(B)$
contains no alternation $abab\ldots$ of length $s$, and such that
the symbols in $T(B)$ make their first appearances in the same order
as they did in $B$, then the resulting sequence \emph{still}
contains no forbidden alternation of length $s+2$.
\end{enumerate}
\end{lemma}

\begin{proof}
Again by induction. Both properties clearly hold if $s=2$, $k=0$, or
$m=1$, so let $s\ge 4$, $k\ge 1$, and $m\ge 2$.

Assume by induction that Properties~\ref{property_no_forbid} and
\ref{property_replace_no_forbid} hold for the sequences $S'$,
$\overline S$, and $S^*$ from which $S^s_k(m)$ is built. We want to
show that these properties hold for $S^s_k(m)$ itself.

We start with Property~\ref{property_no_forbid}. Suppose for a
contradiction that $S^s_k(m)$ contains a forbidden alternation
$abab\ldots$ or $baba\ldots$ of length $s+2$. By
Lemma~\ref{lemma_at_most_ababa}, $a$ and $b$ must have the same
depth (since $s+2 \ge 6$).

If $a$ and $b$ have depth at most $m-1$, then they must belong to
the same copy of $S'$, or else they would not alternate at all. But
this contradicts our inductive assumption on $S'$.

And if $a$ and $b$ have depth $m$ and come from $\widehat S^*$, then
$\widehat S^*$ itself contains a forbidden alternation. But
$\widehat S^*$ is obtained from $S^*$ via block replacements,
exactly as described in Property~\ref{property_replace_no_forbid}.
Thus, the inductive assumption on $S^*$ is contradicted.

In conclusion, $S^s_k(m)$ cannot contain an alternation of length
$s+2$, so it satisfies Property~\ref{property_no_forbid}.

Now we show that $S^s_k(m)$ satisfies
Property~\ref{property_replace_no_forbid}. Suppose for a
contradiction that, after performing a certain set of block
replacements in $S^s_k(m)$, we do get an alternation $abab\ldots$ or
$baba\ldots$ of length $s+2$.

For this to happen, $a$ and $b$ must have appeared together in some
block $B$ of $S^s_k(m)$. (By Lemma~\ref{lemma_easy_properties}, they
do not appear together in more than one block.) Say that $a$
appeared before $b$ in this block. This block was replaced, in the
worst case, by a sequence containing an alternation $abab \ldots$ of
length $s-1$. (Without loss of generality we may assume the
alternation starts with an $a$, since the block replacement
preserves the order of first appearances of the symbols.)

\begin{figure}
\centerline{\includegraphics{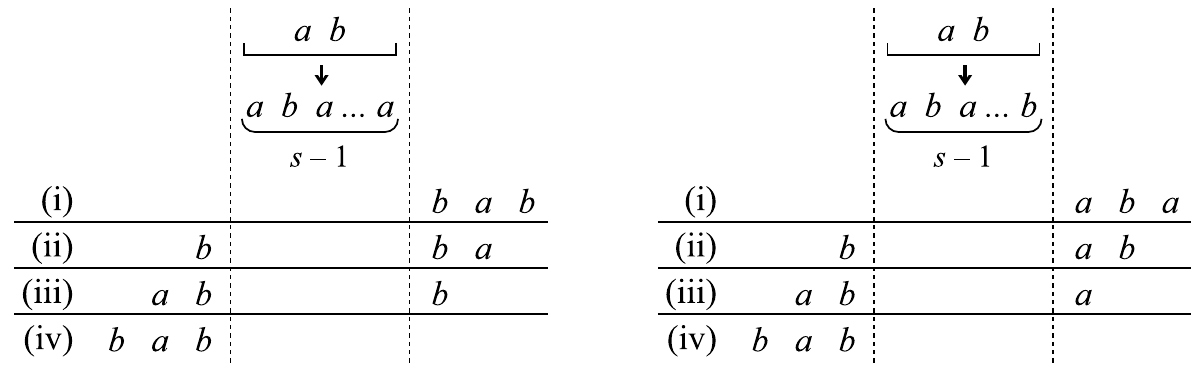}}
\caption{\label{fig_ab_even_odd}The left figure shows the case of
$s$ even. For a forbidden alternation to occur, a pair of symbols
$a$, $b$ in a common block must be replaced by an alternation of
length at most $s-1$, and extended to length $s+2$ by at least three
more symbols $a$, $b$, according to one of four possible cases. In
each case we get a contradiction. The right figure shows the case of
$s$ odd. Here the argument fails, because case (ii) fails to yield a
contradiction.}
\end{figure}

This alternation is extended to length $s+2$ by at least three more
instances of $a$ and $b$ before or after the block $B$, according to
one of four possible cases, as depicted in
Figure~\ref{fig_ab_even_odd} (left).

To see why none of these cases can occur, consider again where the
symbols $a$ and $b$ came from. If $a$ and $b$ came from the same
copy of $S'$, then the same block replacement in $S'$ would also
have generated a forbidden alternation of length $s+2$. This
contradicts our inductive assumption for $S'$.

Further, $a$ and $b$ could not have come from different copies of
$S'$, since then they would not lie together in the same block (and
they would not alternate at all). For a similar reason, they cannot
both come from $\widehat S^*$.

Thus, one symbol---specifically, $a$---must originate from a copy of
$S'$, and the other one---namely, $b$---must originate from
$\widehat S^*$. But all the other instances of $a$ in $S^s_k(m)$, to
the left or right of our block $B$, also come from the same copy of
$S'$. A case analysis shows that in each of the four cases shown in
Figure~\ref{fig_ab_even_odd} (left), this copy of $S'$ received two
copies of $b$ from $\widehat S^*$. (In cases (i) and (ii) there are
two $b$'s surrounded by $a$'s, and in cases (iii) and (iv) there is
a $b$ surrounded by $a$'s, plus another $b$ lying in the same block
as an $a$.) This is impossible according to our construction.
\end{proof}

\begin{remark}
Unfortunately, the above argument depends crucially on $s$ being
even. If we try to make the same argument with $s$ odd, we get the
four cases illustrated in Figure~\ref{fig_ab_even_odd} (right), and
in case (ii) we fail to get a contradiction---we cannot find two
instances of $b$ sent to the same copy of $S'$.
\end{remark}

\subsection{Analysis}\label{subsec_analysis}

Given a fixed even number $s\ge 4$, take the sequences $S^s_k(2)$,
for $k = 0, 1, 2, \ldots$. These are Davenport--Schinzel sequences
of order $s$, in which the multiplicity of the symbols, $\mu_s(k)$,
goes to infinity. Thus, the length of these sequences grows
superlinearly in the number of symbols. We want to derive the exact
relation between these two quantities. For this purpose, we derive
an upper bound on the number of distinct symbols in $S^s_k(2)$.

Let $N^s_k(m) = \left\| S^s_k(m) \right\|$ denote the number of
distinct symbols in $S^s_k(m)$, and let $F^s_k(m)$ be the number of
blocks in $S^s_k(m)$. Then,
\begin{equation}\label{eq_relationNF}
|S^s_k(m)| = \mu_s(k) N^s_k(m) = m F^s_k(m).
\end{equation}
The quantities $N^s_k(m)$ are initialized by
\begin{align*}
N^2_k(m) &= m;\\
N^s_0(m) &= m;\\
N^s_k(1) &= 1.
\end{align*}

To get a recurrence relation for the general case, we analyze the
recursive construction of $S^s_k(m)$. Using the notation there, we
have
\begin{align*}
f &= F^s_k(m-1);\\
g &= N^{s-2}_{k-1}(f);\\
h &= F^s_{k-1}(g) \cdot F^{s-2}_{k-1}(f);\\
N^s_k(m) &= N^s_{k-1}(g) + h\cdot N^s_k(m-1).
\end{align*}
Thus, applying (\ref{eq_relationNF}) three times and then
(\ref{eq_rec_mu}),
\begin{align*}
N^s_k(m) &= N^s_{k-1}(g) + F^s_{k-1}(g)\cdot F^{s-2}_{k-1}(f)\cdot
N^s_k(m-1)\\
&= N^s_{k-1}(g) + {\mu_s(k-1) N^s_{k-1}(g) \over g } \cdot {
\mu_{s-2}(k-1)\cdot g \over f } \cdot{
(m-1) \cdot f \over \mu_s(k) }\\
&= m \cdot N^s_{k-1}(g)\\
&= m\cdot N^s_{k-1} \bigl( N^{s-2}_{k-1} \bigl( F^s_k(m-1) \bigr)
\bigr).
\end{align*}
Since $\mu_s(k) \le 2^{2^k}$ and $m\ge 1$, by (\ref{eq_relationNF})
we have
\begin{equation}\label{eq_bd_F_with_N}
F^s_k(m) \le 2^{2^k} N^s_k(m),
\end{equation}
so
\begin{equation*}
N^s_k(m) \le m\cdot N^s_{k-1}{\left( N^{s-2}_{k-1}{\left( 2^{2^k}
N^s_k(m-1) \right)} \right)}.
\end{equation*}

We now simplify the analysis by getting rid of the dependence on $s$
in the last inequality. For this, we define an Ackermann-like
hierarchy of functions $\widehat A_k(m)$ for $k\ge 0$, $m\ge 1$, by
\begin{equation*}
\widehat A_0(m) = m;
\end{equation*}
and
\begin{equation*}
\widehat A_k(m) =
\begin{cases}
1, & \text{if $m = 1$};\\
m\cdot \widehat A_{k-1}{\left( \widehat A_{k-1}{\left( 2^{2^k}
\widehat A_k(m-1) \right)} \right)}, & \text{otherwise};
\end{cases}
\end{equation*}
for $k\ge 1$ (compare to (\ref{eq_def_Ak})). It follows by induction
that
\begin{equation}\label{eq_N_le_Ahat}
N^s_k(m) \le \widehat A_k(m)
\end{equation}
for all $s$, $k$, and $m$. In Appendix~\ref{app_ack_like} we prove
that
\begin{equation}\label{eq_A_hat_A}
\widehat A_k(m) \le A_{k+1}(2m+4) \qquad \text{for all $k\ge 2$ and
all $m$}.
\end{equation}

Now let us come back to the sequences with which we started this
discussion. Let $T_k = S^s_k(2)$ for $k = 0, 1, 2, \ldots$, and let
$n_k = \|T_k\|$. Then, applying (\ref{eq_N_le_Ahat}),
(\ref{eq_A_hat_A}), and (\ref{eq_another_rec_A}),
\begin{equation*}
n_k = N^s_k(2) \le \widehat A_k(2) \le A_{k+1}(8) \le A_{k+1}\bigl(
A(k+2) \bigr) = A(k+3).
\end{equation*}
Therefore, $k \ge \alpha(n_k) - 3$. Substituting into
(\ref{eq_relationNF}) applying (\ref{eq_bd_mu}), and letting $t =
s/2-1$,
\begin{equation*}
|T_k| = n_k \cdot \mu_s(k) \ge n_k \cdot \mu_s\bigl( \alpha(n_k) - 3
\bigr) \ge n_k \cdot 2^{(1/t!)\alpha(n_k)^t - O{\left(
\alpha(n_k)^{t-1} \right)}}.
\end{equation*}

We have thus achieved the desired lower bound on $\lambda_s(n)$ for
$n$ of the form $n = n_k$. As in Section~\ref{sec_constr_3},
interpolating to intermediate values of $n$ (for $n_k \le n <
n_{k+1}$) is straightforward, and we obtain the desired bound for
all $n$.

\subsection{Advantages over the previous construction}

The construction we just presented follows the same basic idea as
the previous construction of Agarwal et al.~\cite{ASS89, DS_book},
but it has the following advantages:
\begin{itemize}
\item
In our construction each block is just a sequence of $m$ distinct
symbols. In the previous construction each block (there called a
``fan") is of the form $1 2 \ldots m \ldots 2 1$.

\item
In our construction all symbols have the same exact multiplicity.
This greatly simplifies calculations.

\item
In our construction there are no adjacent repeated symbols at the
interface between blocks. (Removing these adjacent repetitions in
the previous construction does not present any serious problem, but
they constitute a small aesthetic blemish.)

\item
The previous construction involves some ``tiny" duplications of
symbols, which our construction does not have. These duplications
are not the cause of the asymptotic growth (and indeed, our
construction works fine without them). This is a potential source of
confusion, especially since these ``tiny" duplications are also
present in the lower-bound construction for order-$3$ sequences, and
in that case they \emph{are} critical.
\end{itemize}

\subsection{Lower bounds for the number of symbols in
almost-DS sequences of even order $s\ge 4$} \label{subsec_ADS_even}

As was the case with the construction of order $3$, the construction
described in this section yields lower bounds for $\N^s_k(m)$, for
$s\ge 4$ even. Again, the idea is to look at the \emph{rows} of the
construction, namely at $S^s_k(m)$ for fixed $s$ and $k$.

\begin{lemma}\label{lemma_ADS_even_lower}
For every fixed even $s\ge 4$ and every $k\ge 4$ we have
\begin{equation*}
\N^s_\mu(x) \ge x \alpha_k(x) \qquad \text{for all large enough
$x$},
\end{equation*}
for some $\mu$ asymptotically of the form
\begin{equation*}
\mu\ge 2^{(1/t!) k^t - O{\left( k^{t-1} \right)}},
\end{equation*}
where $t = s/2 - 1$. Moreover, these lower bounds can be achieved by
actual Davenport--Schinzel sequences.
\end{lemma}

The proof is similar to the proof of Lemma \ref{lemma_ADS_3_lower},
though somewhat simpler, since the blocks in $S^s_k(m)$ have uniform
length. We omit the details.

As before, Lemma~\ref{lemma_ADS_even_lower} automatically yields
lower bounds for $\Nf_{r,s,k}(m)$ for odd $s\ge 5$.

It is an open problem whether the lower bounds for the case $s = 3$
shown above (Section~\ref{subsec_ADS_3}) can be achieved with actual
Davenport--Schinzel sequences (without adjacent repeated symbols),
as was the case here.

\section{Conclusion and open problems}

The bounds for $\lambda_s(n)$ are now tight for every even $s$.
Unfortunately, for odd $s\ge 5$ the problem is still not completely
solved. We believe the new upper bounds for odd $s$ are the true
bounds, simply by analogy to the interval-chain bounds. But the
construction that gives the lower bounds does not seem to work when
$s$ is odd.

Are there other problems that, like interval chains and almost-DS
sequences, satisfy recurrences like Recurrence~\ref{rec_N3km} and
Recurrence~\ref{rec_nskm}? If so, it would be interesting to find
more examples of such problems.

The reason we can unambiguously talk about the coefficient that
multiplies $\alpha(n)$ (e.g., in Theorems~\ref{thm_lambda_new_upper}
and~\ref{thm_lambda3_lower}), despite the fact that there are
several different versions of $\alpha(n)$ in the literature, is that
all these versions differ from one another by at most an
\emph{additive} constant. Thus, the coefficient multiplying
$\alpha(n)$ is not affected. On the other hand, one cannot talk
about the leading coefficient in $\lambda_4(n) = \Theta{\bigl( n
\cdot 2^{\alpha(n)} \bigr)}$, for example, unless a standard
definition of $\alpha(n)$ is agreed upon.

Can our lower-bound construction for $\lambda_3(n)$
(Section~\ref{sec_constr_3}) be realized as the lower envelope of
segments in the plane? If so, it would yield a factor-of-$2$
improvement for this problem as well.

\paragraph{Acknowledgements.}

The ``inverse" problem of almost-DS sequences was raised by my
advisor, Micha Sharir, during a discussion with Haim Kaplan and me.
Recurrence~\ref{rec_N3km} for $\N^3_k(m)$ is also due to Micha, as
well as the proof of Lemma~\ref{lemma_Ex_to_F}. Haim found how an
upper bound for $\N^3_k(m)$ yields an upper bound for
$\lambda_3(n)$; his argument has been generalized in
Lemma~\ref{lemma_psi_to_ADS}.

I also owe thanks to Micha for encouraging me to get into this
intimidating topic in the first place (on account of
Conjecture~\ref{conjecture_lambda}), and for many intensive
discussions. Micha also read carefully several drafts of this paper.

Finally, I wish to thank Martin Klazar for some useful email
correspondence, and to the referees for their careful reading and useful comments.

\appendix

\section{Proof of Klazar's Lemma~\ref{lemma_lamda_3_to_psi_klazar}}
\label{app_lambda_3_klazar}

For completeness, we include here the proof of Klazer's
Lemma~\ref{lemma_lamda_3_to_psi_klazar}. Recall that the claim is
that $\lambda_3(n) \le \psi_3(1 + 2n/\ell, n) + 3n\ell$, where
$\ell\le n$ is a free parameter.

\begin{proof}[Proof of Lemma~\ref{lemma_lamda_3_to_psi_klazar}]
Let $S$ be a maximum-length Davenport--Schinzel sequence of order
$3$ on $n$ distinct symbols. Thus, $|S| = \lambda_3(n)$. Call an
occurrence of a symbol $a$ in $S$ a \emph{terminal occurrence} if it
is the first or last occurrence of $a$ in $S$.

Partition $S$ into blocks $S = S_1 S_2 S_3 \ldots S_m$, where each
$S_i$ starts with a terminal occurrence and contains exactly $\ell$
terminal occurrences (except for $S_m$, which might contain fewer
terminal occurrences). Since $S$ contains $2n$ terminal occurrences,
the number of blocks is $m = \lceil 2n / \ell \rceil \le 1+
2n/\ell$.

For every block $S_i$ and every symbol $a$, let $n_i(a)$ be the
number of occurrences of $a$ in $S_i$. Recall that these occurrences
must be nonadjacent. If $S_i$ contains the first or last occurrence
of $a$ in $S$, we say that $a$ is \emph{terminal in $S_i$};
otherwise, $a$ is \emph{nonterminal in $S_i$}.

Let $\Lambda_i$ be the set of symbols that appear in $S_i$. Let
$\Lambda'_i$ be the subset of these symbols which are terminal in
$S_i$, and let $\Lambda''_i$ be the subset of those which are
nonterminal. Clearly,
\begin{equation*}
|S_i| = \|S_i\| + \sum_{a\in \Lambda_i} \bigl(n_i(a) - 1\bigr).
\end{equation*}

We claim that $n_i(a) \le \ell$ for all $a\in \Lambda_i$. Indeed,
suppose for a contradiction that $n_i(a) \ge \ell + 1$ for some
$a\in \Lambda_i$. Then the occurrences of $a$ in $S_i$ define $\ell$
interior-disjoint, nonempty intervals. But $S_i$ contains at most
$\ell$ terminal occurrences of symbols, one of which is the first
symbol of $S_i$. Therefore, one of the above-mentioned intervals
must be free of terminal occurrences, and so it contains a symbol
$b$ which also appears both before and after the interval. Thus, $S$
contains $babab$, which is a contradiction.

For a similar reason, $S_i$ cannot contain the pattern $aba$ for any
$a,b \in \Lambda''_i$. Therefore, the nonterminal symbols in $S_i$
do not intermingle at all (meaning, for every $a,b\in\Lambda''_i$,
all occurrences of $a$ appear before all occurrences of $b$ or vice
versa). Therefore, the symbols which are nonterminal in $S_i$ define
$\sum_{a\in \Lambda''_i} \bigl(n_i(a) - 1 \bigr)$ interior-disjoint,
nonempty intervals of the form $a\ldots a$ in $S_i$. On the other
hand, the number of such intervals cannot be larger than $\ell-1$
(by an argument similar to the one above). Therefore,
\begin{align*}
|S_i| &= \|S_i\| + \sum_{a\in \Lambda'_i}  \bigl(n_i(a) - 1\bigr) +
\sum_{a\in \Lambda''_i}  \bigl(n_i(a) - 1\bigr)\\
&\le\|S_i\| + (\ell - 1) |\Lambda'_i| + (\ell-1)\\
&\le \|S_i\| + \ell(\ell - 1) + (\ell-1) = \|S_i\| + \ell^2 - 1.
\end{align*}

Now, define a subsequence $S'$ of $S$ by taking just the first
occurrence of each symbol in each $S_i$. Then, $|S'| = \sum_{i=1}^m
\|S_i\|$, and $S'$ is composed of $m$ blocks, each of distinct
symbols. $S'$ might still contain adjacent repeated symbols at the
interface between blocks, but these can be eliminated by deleting at
most $m-1 \le 2n/\ell$ symbols. We get a Davenport--Schinzel
sequence $S''$ which satisfies $|S''| \le \psi_3(m,n)$, and thus
\begin{align*}
\lambda_3(n) = |S| = \sum_{i = 1}^m |S_i| &\le m(\ell^2 -1) +
\sum_{i=1}^m \|S_i\| \\
&\le (1 + 2n/\ell)(\ell^2-1) + \psi_3(m,n) + 2n/\ell \\
&\le \psi_3(1+ 2n/\ell, n) + 3n\ell. \qedhere
\end{align*}
\end{proof}

\section{On the asymptotic growth of some recurrent quantities}
\label{app_growth_constants}

A recurrent feature in this paper are two-parameter quantities given
roughly by $C_{s,k} \approx C_{s-2,k} C_{s,k-1}$, with base cases
$C_{3,k} = \Theta(k)$ and $C_{4,k} = \Theta{\left( 2^k \right)}$.
(Specifically, we have the quantities $P_{s,k}$ and $Q_{s,k}$ in
Section~\ref{sec_psi_old}, and $R_s(d)$ in
Section~\ref{sec_psi_new}. See also $\mu_s(k)$ in
Section~\ref{sec_constr_even}. There are also similar quantities in
\cite{interval_chains}.) In this appendix we give a generic analysis
of the asymptotic growth of such quantities (as a function of $k$
for $s$ fixed).

\begin{lemma}\label{lemma_asymp_common_qtty}
Let $C_{s,k}$ be defined recursively for $s\ge 3$, $k\ge 1$ by
\begin{align*}
C_{3,k} &= \Theta(k);\\
C_{4,k} &= \Theta{\bigl( 2^k \bigr)};\\
C_{s,k} &= C_{s-2,k} C_{s,k-1} + a C_{s-1,k}, \qquad \text{for $s\ge
5$, $k\ge 2$};
\end{align*}
for some implicit constants for $C_{3,k}$ and $C_{4,k}$, some
nonnegative constant $a = a(s)$, and some initial conditions
$C_{s,1}$. Then for every fixed $s\ge 3$ we have
\begin{equation*}
C_{s,k} =
\begin{cases}
2^{(1/t!) k^t \pm O{\left( k^{t-1} \right)}}, & \text{$s$ even};\\
2^{(1/t!) k^t \log_2 k \pm O( k^t )}, & \text{$s$ odd};
\end{cases}
\end{equation*}
where $t = \lfloor (s-2) / 2 \rfloor$.
\end{lemma}

\begin{remark}
The coefficient $a$ of $C_{s-1,k}$ is nonnegative in all
applications of Lemma~\ref{lemma_asymp_common_qtty} we need, so we
do not consider the case where $a$ is negative. And if the
recurrence equation for $C_{s,k}$ had other lower-order terms such
as $a' C_{s-2,k}$ or $a'' C_{s,k-1}$ (for $a'$, $a''$ positive or
negative), they could be handled quite easily, and they would not
affect the asymptotic growth of $C_{s,k}$; therefore, we omit them
for simplicity.
\end{remark}

\begin{proof}[Proof of Lemma~\ref{lemma_asymp_common_qtty}]
Let $s\ge 5$, and assume by induction that $C_{s-1,k}$, $C_{s-2,k}$
have the claimed growth in $k$. Let $c_{s,k} = \log_2 C_{s,k}$.
Then,
\begin{equation*}
c_{s,k} \ge c_{s-2,k} + c_{s,k-1},
\end{equation*}
so
\begin{equation*}
c_{s,k} \ge \sum_{i=2}^k c_{s-2,i}.
\end{equation*}
Using the assumed growth for $c_{s-2,k}$, and bounding the resulting
sum by an integral, we conclude that
\begin{equation*}
c_{s,k} \ge
\begin{cases}
{1\over t!} k^t - O{\bigl( k^{t-1} \bigr)}, & \text{$s$ even};\\
{1\over t!} k^t \log_2 k - O(k^t), & \text{$s$ odd};
\end{cases}
\end{equation*}
implying the lower bound for $C_{s,k}$.

From this lower bound for $C_{s,k}$ it follows that $C_{s,k-1} \ge a
C_{s-1,k}$ for all large enough $k$, and therefore,
\begin{equation*}
C_{s,k} \le 2C_{s-2,k} C_{s,k-1} \qquad \text{for all large enough
$k$}.
\end{equation*}
Thus, the upper bound for $C_{s,k}$ follows similarly.
\end{proof}

\section{Comparing Ackermann-like functions}\label{app_ack_like}

In this appendix we present a general technique for proving that
variants of the Ackermann hierarchy exhibit equivalent rates of
growth. We first give the lemma on which the technique is based, and
then we illustrate the technique by proving that the function
$\widehat A_k(m)$ of Section~\ref{subsec_analysis} satisfies
$\widehat A_k(m) \le A_{k+1}(2m+4)$.

This is basically the same technique as in Appendix B of
\cite{interval_chains}, but rephrased so as to deal with rapidly
growing functions instead of their slowly growing inverses. Our
technique here is also slightly more general than the one in
\cite{interval_chains}.

We consider the following general setting. Suppose $F(n)$ and $G(n)$
are nondecreasing functions that satisfy $F(n), G(n)>n$ for all $n$. Define
functions $F^\circ(n)$, $G^\circ(n)$ by $F^\circ(n) = F^{(n)}(F_0)$,
$G^\circ(n) = G^{(n)}(G_0)$, with some initial conditions $F_0$,
$G_0$. (Recall that $f^{(n)}$ denotes the $n$-fold composition of
$f$.)

We want to prove that $F^\circ(n) \le G^\circ(dn + c)$ for some
constants $d$ and $c$. The following lemma gives a sufficient
condition for this.

\begin{lemma}\label{lemma_FG_delta}
Let $F(n)$, $G(n)$, $F^\circ(n)$, $G^\circ(n)$ be functions as given
above. Suppose there exists an integer $d$ and a function
$\delta(n)$ such that
\begin{align}
n &\le \delta(n), \label{eq_app_n_delta_n} \\
\delta(F(n)) &\le G^{(d)}(\delta(n)), \label{eq_app_deltaFG}
\end{align}
for all $n\ge 1$. Then $F^\circ(n) \le G^\circ(dn+c)$ for a constant
$c$ large enough that
\begin{equation}\label{eq_condition_c}
\delta(F_0) \le G^\circ(c).
\end{equation}
\end{lemma}

\begin{proof}
Applying (\ref{eq_app_n_delta_n}), then (\ref{eq_app_deltaFG}) $n$
times, and then (\ref{eq_condition_c}),
\begin{align*}
F^\circ(n) = F^{(n)}(F_0)  \le \delta{\bigl( F^{(n)}(F_0)
\bigr)} &\le G^{(dn)}\bigl( \delta(F_0) \bigr)\\
&\le G^{(dn)}\bigl(G^\circ (c) \bigr) = G^{(dn+c)}(G_0) =
G^\circ(dn+c). \qedhere
\end{align*}
\end{proof}

Now let us apply this technique to the task at hand.

\begin{lemma}
Let $\widehat A_k(m)$ be given by
\begin{equation*}
\widehat A_0(m) = m, \qquad \text{for $m\ge 1$};
\end{equation*}
and
\begin{equation}\label{eq_rec_Ahat}
\widehat A_k(m) =
\begin{cases}
1, & \text{if $m = 1$};\\
m\cdot \widehat A_{k-1}{\left( \widehat A_{k-1}{\left( 2^{2^k}
\widehat A_k(m-1) \right)} \right)}, & \text{otherwise};
\end{cases}
\end{equation}
for $k\ge 1$. Then $\widehat A_k(m) \le A_{k+1}(2m+4)$ for all $k\ge
2$ and all $m$.
\end{lemma}

\begin{proof}
We start by noting that
\begin{equation}\label{eq_Ahat1}
\widehat A_1(m) = 2^{2m-2} m! \le 2^{m^2}.
\end{equation}
Unfortunately the recurrence (\ref{eq_rec_Ahat}) does not fit the
general setting of Lemma~\ref{lemma_FG_delta} because of the factor
$m$ in it. But it is not hard to show that
\begin{equation*}
\widehat A_k(m) \le \widehat A_{k-1}{\left( \widehat A_{k-1}{\left(
2^{2^k} \widehat A_k(m-1) \right)} \right)}^2 \qquad \text{for $m\ge
2$},
\end{equation*}
so we will use this recurrence instead (the penalty we pay is
minimal). We are going to apply Lemma~\ref{lemma_FG_delta} with $d =
2$, with
\begin{align}
F(m) &= \widehat A_{k-1}{\left( \widehat A_{k-1}{\left( 2^{2^k} m
\right)} \right)}^2, \label{eq_rec_Fm}\\
G(m) &= A_k(m), \nonumber
\end{align}
and with the initial conditions $F_0 = G_0 = 1$. Thus,
\begin{align*}
F^\circ(m) &\ge \widehat A_k(m),\\
G^\circ(m) &= A_{k+1}(m).
\end{align*}

Let us start with the case $k = 2$. In this case we have, by
(\ref{eq_Ahat1}),
\begin{align*}
F(m) &= \widehat A_1(\widehat A_1(16m))^2 \le 2^{2^{512m^2+1}},\\
G(m) &= 2^m.
\end{align*}
Then an appropriate choice of $\delta$ is $\delta(m) = 600 m^3$,
since
\begin{equation*}
\delta(F(m)) \le \delta{\left( 2^{2^{512m^2 + 1}} \right)} = 600\cdot
2^{3\cdot2^{512 m^2+1}} \le 2^{2^{600 m^3}} = G(G(\delta(m)))
\end{equation*}
for all $m\ge 1$, and so $\delta$ satisfies (\ref{eq_app_deltaFG}).
Further, it is enough to take $c = 4$ in (\ref{eq_condition_c}),
since
\begin{equation*}
G^\circ(4) = 2^{2^{2^2}} \ge 515 = \delta(F_0).
\end{equation*}
We conclude that $\widehat A_2(m) \le A_3(2m+4)$.

Now we deal with the general case $k\ge 3$. Suppose by induction
that $\widehat A_{k-1}(m) \le A_k(2m+4)$. Substituting this into
(\ref{eq_rec_Fm}),
\begin{equation*}
F(m) \le A_k{\left( 2 A_k{\left( 2^{2^k+1} m + 4 \right)} + 4
\right)}^2.
\end{equation*}
Now it is easy to see that taking $\delta(m) = 2^{2^k+1}m + 5$
guarantees that
\begin{equation*}
\delta(F(m)) \le A_k(A_k(\delta(m))) = G(G(\delta(m)))
\end{equation*}
for all $m\ge 1$. Furthermore, we have
\begin{equation*}
G^\circ(4) = A_{k+1}(4) > 2^{2^k+1}+5 = \delta(F_0).
\end{equation*}
We conclude that $\widehat A_k(m) \le A_{k+1}(2m+4)$, as desired.
\end{proof}

\end{document}